\documentclass[letterpaper,11pt]{article}

\usepackage[T1]{fontenc}
\usepackage[utf8]{inputenc}
\usepackage{lmodern} %
\usepackage{authblk}
\usepackage[bookmarks,colorlinks,breaklinks]{hyperref}
\hypersetup{urlcolor=blue, colorlinks=true, citecolor=green!50!black, linkcolor=blue}
\usepackage{multirow}
\usepackage[T1]{fontenc} %
\usepackage{amsmath}
\usepackage{amsfonts}
\usepackage{amssymb}
\usepackage{amsthm}
\usepackage{thmtools}
\usepackage{pgffor}
\usepackage{xcolor}
\usepackage[lined,boxed,ruled,norelsize,algo2e,linesnumbered,noend]{algorithm2e}
\usepackage{cleveref}
\usepackage{graphicx}
\usepackage{geometry}
\usepackage{enumitem}
\usepackage{url}
\usepackage{wrapfig}
\ifdefined\DEBUG
	\usepackage{showlabels}
    \usepackage{lineno}
    \linenumbers

\else
\fi

\declaretheorem[numberwithin=section,refname={Theorem,Theorems},Refname={Theorem,Theorems}]{theorem}
\declaretheorem[numberlike=theorem]{lemma}

\declaretheorem[numberlike=theorem]{corollary}
\declaretheorem[numberlike=theorem]{definition}

\newcommand{\F}{\mathbb{F}}

\newcommand{\Z}{\mathbb{Z}}

\renewcommand{\P}{\mathbb{P}}

\renewcommand{\tilde}{\widetilde}

\newcommand{\tO}{\tilde{O}}

\DeclareMathOperator*{\MM}{MM}

\newcommand{\poly}{\operatorname{poly}}
\newcommand{\polylog}{\operatorname{polylog}}
\newcommand{\rank}{\operatorname{rank}}
\renewcommand{\mod}{\operatorname{mod}}

\foreach \x in {A,...,Z}{%
	\expandafter\xdef\csname m\x\endcsname{\noexpand\mathbf{\x}}
}

\foreach \x in {A,...,Z}{%
	\expandafter\xdef\csname om\x\endcsname{\noexpand\overline{\noexpand\mathbf{\x}}}
}

\foreach \x in {A,...,Z}{%
	\expandafter\xdef\csname c\x\endcsname{\noexpand\mathcal{\x}}
}

\newcommand{\otau}{\noexpand{\overline{\tau}}}
\newcommand{\omu}{\noexpand{\overline{\mu}}}
\newcommand{\osigma}{\noexpand{\overline{\sigma}}}%
\newcommand{\os}{\noexpand{\overline{s}}}
\newcommand{\ot}{\noexpand{\overline{t}}}
\newcommand{\og}{\noexpand{\overline{g}}}
\newcommand{\ou}{\noexpand{\overline{u}}}
\newcommand{\ov}{\noexpand{\overline{v}}}
\newcommand{\ow}{\noexpand{\overline{w}}}
\newcommand{\ox}{\noexpand{\overline{x}}}
\newcommand{\oy}{\noexpand{\overline{y}}}
\newcommand{\oz}{\noexpand{\overline{z}}}
\newcommand{\oc}{\noexpand{\overline{c}}}

\newcommand{\init}{\mathrm{(init)}}%

\makeatletter
\renewcommand{\paragraph}{%
	\@startsection{paragraph}{4}%
	{\z@}{1.25ex \@plus 1ex \@minus .2ex}{-1em}%
	{\normalfont\normalsize\bfseries}%
}
\makeatother

\title{
Dynamic Distances Without Matrix Inverse
}
\author{
Jan van den Brand\thanks{Georgia Institute of Technology. Supported by NSF-2504994 and NSF-2338816.}\quad\quad
Adam Karczmarz\thanks{University of Warsaw. Supported by the National Science Centre (NCN) grant no.~2022/47/D/ST6/02184.}\quad\quad
Aryan Patel\thanks{Georgia Institute of Technology}}
\date{}

\begin{document}

\pagenumbering{roman}
\maketitle

\begin{abstract}

We study the fully dynamic exact distance problem in unweighted directed graphs, where the graph undergoes edge insertions and deletions, and the distance between two query vertices must be reported.
Under the combinatorial Boolean matrix multiplication conjecture, any subquadratic-time algorithm for this problem must use algebraic techniques.
All prior work uses the polynomial matrix inverse as its algebraic tool, but within subquadratic time, this technique can only maintain distances up to a small-polynomial hop bound.
So larger distances must be reassembled from short-range pieces by a combinatorial layer. 
This hybrid algebraic-combinatorial construction has led to a substantial slowdown for dynamic distances compared to other dynamic graph problems relying on algebraic techniques.

We give the first subquadratic-time dynamic distance algorithm built on dynamically maintaining a different algebraic object than the matrix inverse.
By maintaining the \emph{Frobenius form} $\mA=\mK\mC\mK^{-1}$ of a dynamic generic (i.e., cyclic/non-derogatory) 
matrix $\mA$, we can query all the powers $\mA^\ell = \mK \mC^\ell \mK^{-1}$ efficiently. For adjacency matrix~$\mA$, the smallest $\ell$ with $(\mA^\ell)_{s,t}\neq 0$ is the $st$-distance in the associated graph.
So, unlike dynamic matrix inverse, we get distance information directly without hop truncation, and we no longer need to perform any combinatorial reassembly.
This purely algebraic approach to dynamic distances leads to the fastest dynamic distance algorithm with $O(n^{1.58})$ update time.

The earlier use of the Frobenius form maintained the dense change of basis $\mK$ explicitly and required $\tilde O (n^2)$ time per update \cite{KarczmarzS23}. 
Our key observation is that this change of basis is highly structured, which lets us maintain it implicitly.
This gives the first subquadratic fully dynamic algorithm for the Frobenius normal form and characteristic polynomial for generic matrices.

\end{abstract}

\clearpage
\tableofcontents

\clearpage
\pagenumbering{arabic}

\section{Introduction}

A dynamic graph algorithm maintains some property of a graph that undergoes a sequence of edge insertions and deletions, with the goal of processing each update far faster than recomputing the property from scratch. 
An algorithm is called incremental if it supports only insertions, decremental if it supports only deletions, and fully dynamic if it supports both.

Among dynamic graph problems, maintaining shortest paths is one of the most extensively studied, with a long history spanning the incremental, decremental, and fully dynamic settings in the 
$st$/arbitrary-pair query, e.g., \cite{RodittyZ11,abs-2408-14406,BrandNS19,BergamaschiHGWW21,KarczmarzS23,AlokhinaB24, BrandFN22,Brand0KLMGS24,Chuzhoy21,DoryFNV22,ChuzhoyZ23,ForsterGNS23,KyngMG24,ForsterNG23,HaeuplerLS24,Chechik18,LackiN22}, single-source, e.g., \cite{BrandK23,EvenS81,BrandN19,GutenbergWW20,KyngMG22,Liu26,BernsteinGW20,BernsteinGS21,BernsteinGS20,HenzingerKN14}, and all-pairs  \cite{King99, AusielloIMN91, BaswanaHS07, DemetrescuI04, Thorup04, AbrahamCK17, ChechikZ23, GutenbergW20b,Mao24a,EvaldFGW21,BrandN19,Bernstein16,KarczmarzL20} variants.
We focus on the fully dynamic setting on unweighted directed graphs undergoing edge updates and exact distance queries between arbitrary pairs of vertices. In particular, we aim for data structures with subquadratic worst-case update and query time. 
Subquadratic time is the natural benchmark, since a single breadth-first search on an $n$-vertex graph recomputes distances from scratch in $O(n^2)$ time, and the motivation of dynamic algorithms is to beat na\"ive recomputation. %

This quadratic barrier is natural and likely cannot be beaten with combinatorial techniques. Abboud and Vassilevska Williams \cite{AbboudW14} showed that under the combinatorial Boolean Matrix Multiplication conjecture, no combinatorial fully dynamic algorithm can achieve $n^{2-\epsilon}$ update and query time.
The only known approach to achieving subquadratic update time therefore runs through fast matrix multiplication, treating the graph as an algebraic object.
The question is then not \emph{whether} to be algebraic, but \emph{which} algebraic paradigm to maintain. 

\paragraph{Dynamic Matrix Inverse}
The algebraic approach traces back to the path-counting algorithms of King-Sagert \cite{KingS02} and Demetrescu-Italiano \cite{DemetrescuI05} for reachability, and was refined by Sankowski, who reduced dynamic reachability \cite{Sankowski04}, matching \cite{Sankowski07}, and distances \cite{Sankowski05} to a single algebraic primitive: \emph{the dynamic matrix inverse}.
Every subsequent algebraic distance algorithm is built on this matrix inversion approach \cite{BrandNS19,BrandN19,BergamaschiHGWW21,BrandFN22,BrandK23,KarczmarzS23,AlokhinaB24}.

However, maintaining distances algebraically is structurally quite different from the other dynamic graph problems.
For most dynamic algebraic graph algorithms, the maintained algebraic object is already the answer. The rank of the Tutte-matrix corresponds to the size of a maximum matching. Whether $(\mI-x\mA)^{-1}_{s,t}\neq 0$ for adjacency matrix $\mA$ corresponds to $st$-reachability.
Dynamic rank and matrix inverse are purely algebraic problems, and their dynamic algorithms are correspondingly efficient.
Dynamic distance is different in kind and requires the combination of both algebraic and combinatorial techniques, resulting in hybrid algorithms.
When solving distances via matrix inversion, the matrix inverse\footnote{One inverts a \emph{polynomial} matrix modulo $x^{h+1}$, i.e., each entry of the matrix inverse is a polynomial of degree $h$.} only returns distances up to a hop bound $h \ll n$.
To extend the algorithm to larger distances, one must reassemble hop-bounded path pieces. This reassembly is the combinatorial layer that the entire line of work makes use of. 
The most prevalent combinatorial method is the hitting-set trick~\cite{UllmanY91}\footnote{The same idea is used in static algebraic all-pairs shortest paths algorithms~\cite{AlonGM97,Zwick02} and numerous combinatorial dynamic shortest paths algorithms, e.g.,~\cite{King99,BaswanaHS07,AbrahamCK17}.}.
Trading off between the algebraic layer (inefficient for \emph{large} hop-bound $h$) and the combinatorial layer (inefficient for \emph{small} hop-bound $h$) results in a substantial slowdown for dynamic distances compared to the other dynamic graph problems.
The algebraic-combinatorial hybrid can be made efficient, but only by relaxing the problem to approximation and undirected graphs. For example, using spanners or emulators for the combinatorial layer \cite{BergamaschiHGWW21,BrandFN22} is efficient but these results are $(1+\epsilon)$-approximate and hold only for undirected graphs. 
Needing an additional combinatorial layer thus always weakened the algorithm substantially, either slowing it down or confining it to a much weaker model.

\paragraph{Dynamic Frobenius Form}
We completely remove the combinatorial part by maintaining a different algebraic primitive than matrix inversion.
The Frobenius normal form (also called the rational canonical form) is a type of matrix decomposition where we express a given matrix $\mA$ as $\mA = \mK \mC \mK^{-1}$. Here, the matrix $\mC$ has a special structure that allows computing $\mC^\ell$ efficiently, allowing us to efficiently query entries of $\mA^\ell=\mK\mC^\ell\mK^{-1}$. To see why this is useful for distance problems, recall that for adjacency matrix $\mA$ of an unweighted graph, $(\mA^\ell)_{s,t}$ equals the number of $st$-walks using $\ell$ steps, so the smallest $\ell$ with $(\mA^\ell)_{s,t}\neq 0$ is the $st$-distance. Hence, a dynamic algorithm maintaining the Frobenius form implies a dynamic distance algorithm for graphs.

This Frobenius form has been reached for before in static~\cite{SankowskiW19,KarczmarzS23} and dynamic~\cite{KarczmarzS23} distance problems but never maintained in subquadratic time. 
Every prior algebraic data structure that achieves subquadratic update time for exact distance uses dynamic inverse maintenance to some extent.
Even the one work \cite{KarczmarzS23} that turns to the Frobenius form does not escape the inverse, but rather shows how firmly it holds.
In the state-of-the-art data structure, \cite[Section~9.1]{KarczmarzS23} uses the Frobenius form only as a static subroutine, accelerating periodic rebuilds of the full polynomial matrix inverse. %
Their second data structure~\cite[Section~9.2]{KarczmarzS23} that avoids full matrix inverse maintenance leverages Frobenius form recomputation (from scratch) combined with matrix inverse arguments for maintaining selected entries of the $h$ first powers of the adjacency matrix~$\mA$.
Summing up, for the current subquadratic distance algorithms, the Frobenius form never becomes the dynamic object but merely a scaffolding inside the matrix inverse machinery, and is always a part of an algebraic-combinatorial hybrid algorithm.

For distance problems, it suffices to focus on Frobenius forms of so-called %
\emph{generic} matrices\footnote{Also called cyclic matrices. Genericity means that some vector $v$ exists where $v,\mA v,\mA^2 v,...,\mA^{n-1} v$ form a basis.
It is not a strong restriction as adjacency matrices can be made generic via randomization (\Cref{thm:weighted-adjacency-matrix}).}. Generic matrices allow for a particularly simple shape of $\mC$ and a transform $\mK$ built from iterates~$\mA^i v$ for some $v$, see, e.g.,~\cite{Keller-Gehrig85}.
The only prior approaches for maintaining the Frobenius form in a dynamic setting resulted in $\tilde O(n^2)$ update time \cite{KarczmarzS23,FrandsenS11} in the generic case, and even higher bounds in the non-generic case~\cite{FrandsenS11}.
The reason is structural: 
an update to matrix~$\mA$ produces a new form $\mK'\mC'\mK'^{-1}$ 
and the new basis-change matrices $\mK',\mK'^{-1}$ are dense $n\times n$ matrices that can differ from $\mK,\mK^{-1}$ significantly\footnote{
We have $\mK' = \mK\mT$ for some transform $\mT$ and original basis-change $\mK$. However, $\mT$ is a dense full-rank matrix, so implicit maintenance of $\mK'$ by explicitly storing $\mK$ and $\mT$ does not work either.}, 
so naively the form cannot be repaired in subquadratic time. 
This raises the question whether a truly dynamic Frobenius form, 
fast enough to beat the inverse, is achievable at all, or whether the dense update to change-of-basis $\mK$ is a barrier.

We show that maintenance of the Frobenius form in subquadratic time is possible.
We prove that the update to $\mK'=\mK\mT$ for some transform matrix $\mT$ is a so-called ``Toeplitz-like'' matrix~\cite{kailath1979displacement, kailath1995displacement}. Toeplitz-like matrices $\mT$ allow implicit representation via smaller $n\times k$ ($k \ll n$) dimensional matrices $\mU,\mV$, called generators. By maintaining these generators of $\mT$ explicitly, we implicitly maintain $\mK' = \mK\mT$ and thus the Frobenius form.
This allows us to achieve $O(n^{1.58})$ update time and query time\footnote{There is a tradeoff parameter that can be chosen to reduce query complexity.} for dynamic distances, resulting in the fastest algorithm yet.
Unlike all prior work, our improvement does not come from a better dynamic matrix inverse data structure or a better tradeoff with the combinatorial layer. Instead, we maintain a completely different algebraic object that yields distances directly, removing the combinatorial layer entirely.

\subsection{Our Results}
\label{sec:intro:results}

As our algorithms are algebraic and use fast matrix multiplication, we use $\MM(m,n,k)$ to denote the complexity of multiplying an $m\times n$ matrix by an $n\times k$ matrix. 
We state our results both in generality (using $\MM(\cdot)$) and for current upper bounds on fast rectangular matrix multiplication~\cite{AlmanDWXXZ25,WilliamsXXZ24}.
We use $\tilde O(\cdot)$ to hide $\polylog(n)$.

\paragraph{Dynamic Distance} Our main result is the following dynamic distance algorithm.

\begin{restatable}{theorem}{dynamicdistances}\label{thm:dynamic-distances}
    Let $G=(V,E)$ be a directed unweighted graph.
    There is a Monte Carlo randomized dynamic data structure maintaining $G$ subject to edge insertions and deletions, and supporting \linebreak $st$-distance queries for any $s,t\in V$. For a tradeoff parameter $1\le k\le n$, the update time is $\tilde O(nk + \MM(n,k,k)n/k^2)$ and query time is $\tilde O(nk)$. The produced answers are correct with high probability\footnote{I.e., with probability at least $1-1/n^\tau$ over $q=\poly(n)$ updates for a chosen constant $\tau>0$ (dependent on $q$).}.
    The data structure uses $O(n^2)$ space.
\end{restatable}

\begin{wrapfigure}{r}{0.4\textwidth}
    \vspace{-28pt}
    \begin{tabular}{c|c|c|c}
        Reference & Update & Query & Space \\
        \hline
        \cite{Sankowski05} & $n^{1.896}$ & $n^{1.528}$ & $n^{2.37}$  \\
        \cite{BrandNS19} & $n^{1.764}$ & $n^{1.764}$ & $n^{2.24}$\\
        \cite{KarczmarzS23} & $n^{1.673}$ & $n^{1.673}$ & $n^{2.33}$ \\
        \textbf{Ours} & $n^{1.58}$ & $n^{1.58}$ & $n^2$\\
        \textbf{Ours} & $n^{1.626}$ & $n^{1.5}$ & $n^2$\\
    \end{tabular}
    \caption{Subquadratic upper bounds for exact distances \label{fig:complexities}}
    \vspace{-10pt}
\end{wrapfigure}

Choosing the parameter $k\approx n^{0.58}$ to minimize the update time results in $O(n^{1.58})$ update and query time \cite{Complexity}.
This improves upon all prior subquadratic update times; see \Cref{fig:complexities}.
We remark that \cite{Sankowski05} has comparably small query time but our update-query tradeoff is better; e.g., we can choose a different $k$ such as $n^{0.5}$ and get $O(n^{1.626})$ update and $O(n^{1.5})$ query time, beating both the update and query\footnote{In \cite{Sankowski05} there is a typo, leading to an incorrectly stated $n^{1.288}$ query time. On page~469, the equation $O(n^{1-\mu+\epsilon}+n^{2-2\mu})$ should have been $O(n^{1-\mu+\epsilon}\cdot n^\mu+n^{2-2\mu})$. The additional $\cdot n^\mu$ overhead for a polynomial matrix of degree $n^\mu$ was missed.} times of \cite{Sankowski05}.

In addition to being much faster than prior work, an additional benefit of our purely algebraic approach is improved space complexity.
\Cref{thm:dynamic-distances} uses $O(n^2)$ space (i.e., linear for dense graphs) whereas the prior hybrid data structures used super-quadratic space $\Theta(n^2h)$ where $h$ is their hop-bound after which longer paths are reassembled via the combinatorial layer.

The Frobenius form approach to distance problems is tailored to unweighted digraphs.
However, a simple reduction (see, e.g.,~\cite{KarczmarzS23}) shows that dynamic distances with weights $\{1,\ldots,W\}$ can be reduced to dynamic distances in unweighted digraphs with $O(nW)$ vertices.
Consequently,~\Cref{thm:dynamic-distances} yields $O((Wn)^{1.58})$ update and query time in that case.
It should be noted that the combination of polynomial matrix inverse and hitting sets, as used in, e.g., \cite{Sankowski05,BrandNS19,BergamaschiHGWW21}, can still be advantageous for graphs with small nonnegative integer weights: unlike the reduction above, it can handle zero-weight edges, %
and the obtained dependence on $W$ can be linear (but is still pseudo-polynomial).

We obtain our dynamic distance result~(\Cref{thm:dynamic-distances}) via the following algebraic data structure providing fast access to powers of a dynamic matrix, which may be of independent interest.
In the following, let $\F$ be a field of size $\geq n^\gamma$, where $\gamma$ is a sufficiently large constant controlling success probability.

\begin{restatable}{theorem}{thmdynamicmatrixpower}\label{thm:dynamic-powers}
    Let $\mA\in\F^{n\times n}$.
    There is a Monte Carlo randomized dynamic data structure maintaining $\mA$ subject to single-element updates and supporting queries that return, for any $i,j$, all $(\mA^\ell)_{i,j}$ for $\ell=1,...,n$. Provided the matrix remains generic throughout all updates, the update time is $\tilde O(nk + \MM(n,k,k)n/k^2)$ and query time is $\tilde O(nk)$, where $k\in [1,n]$ can be chosen freely. The produced answers are correct with high probability.
    The data structure uses $O(n^2)$ space.
\end{restatable}

\paragraph{Low-Rank Update to Frobenius Form}

Given the inverse $\mA^{-1}$, the Woodbury identity\footnote{The Woodbury identity states $(\mA+\mU\mV^\top)^{-1} = \mA^{-1} - \mA^{-1}\mU(\mI+\mV^\top\mA^{-1}\mU)^{-1}\mV^\top\mA^{-1}$.}
computes the inverse of $\mA+\mU\mV^\top$ for $n\times k$ sized $\mU,\mV$ in $O(\MM(n,k,n))$ time.
This is also referred to as a rank-$k$ update since we add a rank-$k$ matrix $\mU\mV^\top$ to $\mA$, and it can be interpreted as a single batch of $k$ smaller updates.
This is a crucial algorithmic tool in dynamic graph algorithms \cite{DemetrescuI05,Sankowski05,Sankowski07,BrandNS19,BergamaschiHGWW21,KarczmarzS23,BrandK23,KarczmarzMS22,AlokhinaB24}, failure sensitive distance oracles \cite{BrandS19,KarczmarzS23,GuR21}, and static optimization algorithms (e.g., \cite{LeeS15,Brand20,CohenLS21,JiangKLPS20,Brand21}).
No equivalent update is known for the Frobenius form. 
The only route to updating it under a rank-$k$ change was to apply $k$ separate rank-$1$ updates in sequence without fast matrix multiplication, each costing $\tilde O(n^2)$, for a total of $\tilde O(kn^2)$.
In particular, for any $k\ge n^{0.372}$ one may as well just recompute the Frobenius form from scratch in $\tO(n^\omega)$ time\footnote{Here, $\omega$ is the matrix multiplication exponent, i.e., $\MM(n,n,n)=O(n^\omega)$. It is known that $\omega\leq 2.372$~\cite{AlmanDWXXZ25}.}~\cite{Keller-Gehrig85}.
Indeed, this was a path taken in~\cite[Section~9.2]{KarczmarzS23}, and a more efficient rank-$k$ update algorithm could immediately improve their dynamic distance data structure (though not as much as our \Cref{thm:dynamic-distances}). We show that a non-trivial rank-$k$ update is indeed possible, and the algorithm may be of independent interest.

\begin{theorem}[Rank-$k$ update to Frobenius form]\label{thm:lowrankfrobenius}
    Given are a Frobenius form $\mA=\mK\mC\mK^{-1}$ of a generic matrix $\mA\in\F^{n\times n}$ and a rank-$k$ update $\mA+\mU\mV^\top$ such that $\mA+\sum_{i=1}^\ell \mU e_i e_i^\top \mV^\top$ (i.e., using only the first $\ell$ columns of $\mU,\mV$) is generic for all $\ell=1,...,k$.
    There is a Monte Carlo randomized algorithm computing w.h.p.~the Frobenius form $\mA+\mU\mV^\top = \mK'\mC'\mK'^{-1}$ in \linebreak $\tilde O(\MM(k,k,n)n/k)$ operations.
\end{theorem}
Note that the above bound is $\tO(n^2k^{\omega-2})$ in particular, so it improves upon the aforementioned naive bound $\tO(\min(kn^2,n^\omega))$ for any $k=n^{\alpha}$ with $\alpha\in (0,1)$ (unless $\omega=2$). 

We remark that the techniques for computing a single rank-$k$ update extend to sensitivity distance oracles as studied in, e.g., \cite{BernsteinK09,DuanP09,WeimannY13,DuanZ17,BrandS19,GrandoniW20,ChechikC20,GuR21,Ren22,KarczmarzS23,BiloCCCFS24,DeyG24}.
These are data structures that, after preprocessing an input graph, receive a single batch of edge updates (failures) and a vertex-pair query $(s,t)$, and then answer the $st$-distance subject to the edge failures. 
A batch of $k$ edge updates/failures can be modeled as a rank-$k$ update to the adjacency matrix. Using the techniques that yield \Cref{thm:lowrankfrobenius} in such a way that they support query access rather than explicitly computing the entire Frobenius form, we obtain the following sensitivity distance oracle for directed graphs and edge failure batches.

\begin{restatable}{theorem}{thmsensitivity}\label{thm:sensitivity-distances}
Let $G$ be a directed unweighted graph.
    We can preprocess $G=(V,E)$ in $\tO(n^\omega)$ time to enable the following.
    Given a batch of $k\leq n$ failed edges $F\subseteq E$, we can process $F$ in $\tilde O(k^{\omega-1}n)$ time, so that queries about $st$-distances in $G\setminus F$ are supported in $\tO(nk)$ time for any $s,t\in V$.
    The produced answers are correct with high probability.
    The data structure uses $O(n^2)$ space.
\end{restatable}
This matches the state-of-the-art bounds for large failures $F$ obtained by~\cite{KarczmarzS23}, who used a more complex hybrid algebraic-combinatorial data structure relying on dynamic matrix inverse methods and the small-hop/large-hop distinction.

\paragraph{Dynamic Characteristic Polynomial}

The characteristic polynomial of a matrix $\mA$ is \linebreak ${\det(x\mI-\mA)}$. It is computable in~$\tO(n^\omega)$ time~\cite{Keller-Gehrig85}. Prior work \cite{FrandsenS11} gave a dynamic algorithm maintaining the characteristic polynomial of generic matrices under rank-1 updates (i.e., $\mA \gets \mA + uv^\top$) in $\tilde O(n^2)$ time.
They also provided an $\Omega(n^2)$ lower bound for rank-1 updates. When the updates change only single entries of $\mA$, then the lower bound becomes $\Omega(n)$. It was explicitly asked in \cite{FrandsenS11} if subquadratic update time is possible when supporting only entry updates. We answer this question positively for generic matrices.
The Frobenius form and the characteristic polynomial are closely related (see \Cref{def:frobenius} for details). In particular, one can read the polynomial directly from the Frobenius form, leading to \Cref{thm:dynamiccharpoly}.

\begin{theorem}\label{thm:dynamiccharpoly}
Let $\mA\in\F^{n\times n}$.
    There is a Monte Carlo randomized dynamic algorithm that maintains (with high probability)~the characteristic polynomial of $\mA$ subject to element updates, provided that $\mA$ remains generic. Each update takes $\tilde O(\min_k\{nk + \MM(n,k,k)n/k^2\}) = O(n^{1.58})$ operations.
\end{theorem}

Our algebraic techniques currently apply to generic (cyclic) matrices, which suffice for the distance application and permit simpler and more elegant algorithms. At the same time, we believe our results constitute the essential first step towards dynamic Frobenius form and characteristic polynomial algorithms for general matrices over fields.

\subsection{Further Related Work}

Our work focuses on exact distances, which likely require algebraic techniques \cite{AbboudW14}. This also holds for $(1+\epsilon)$-approximate distances. However, when allowing for large $\ge2$ approximation errors on undirected graphs, then one can use purely combinatorial methods. E.g., subquadratic time is achieved by maintaining a dynamic spanner such as \cite{BernsteinFH21,AbrahamDKKP16} and running BFS to answer queries. There are also other faster dynamic distance algorithms for large approximation errors on undirected graphs, such as \cite{AbrahamDKKP16,HaeuplerLS24,%
ForsterGH21,%
AbrahamCT14,%
ForsterGNS23,%
ChuzhoyZ23,%
Bernstein09,%
KyngMG24}.

On the lower bound side, there is an $\Omega(n)$ conditional lower bound for reachability and distances, conditional on the OMv-conjecture \cite{HenzingerKNS15}.
Conditional on triangle detection, there is also an $\Omega(n^{\omega-1})$ lower bound \cite{AbboudW14}.
Our work focuses on unweighted graphs, because subquadratic-time dynamic algorithms for weighted\footnote{Allowing integral weights in $[1,n^c]$ for some sufficiently large $c>0$.} graphs are ruled out by the APSP conjecture~\cite{AbboudW14,RodittyZ11}.

\section{Preliminaries}

\paragraph{Fields and Complexity}
Let $\F$ denote a finite field. All our algebraic lemmas for $n\times n$ matrices undergoing $q=\poly(n)$ updates state their complexity as the number of arithmetic operations performed over $\F$ of size $\geq n^\gamma$, where $\gamma> 3+\log_n{q}$ is a constant controlling the success probability. For the graph application, we will use the finite field $\Z_p$ for a prime $p = \poly(n)$. In the standard Word-RAM model, each arithmetic operation over $\Z_p$ for $p=\poly(n)$ takes $O(1)$~time.

For the matrix product of an $n\times m$ and $m\times k$ matrix, we write $\MM(n,m,k)$ for the complexity of matrix multiplication.
Here we have $\MM(n,m,k)=\MM(n,k,m)=...=\MM(k,m,n)$, i.e., the arguments can be permuted, and $\MM(ab,c,d) \le \MM(a,c,d)\cdot b$ because one can split a matrix product into several smaller ones.
Precise upper bounds are given in \cite{AlmanDWXXZ25,WilliamsXXZ24}.

We write $\F[x]$ for the ring of polynomials with coefficients in $\F$. Two polynomials $f,g\in\F[x]$ of degree $\le d$ can be multiplied in $\tO(d)$ operations via fast Fourier transform-type algorithms~\cite{CK91}. Other polynomial operations (addition, subtraction, division) are all bounded by $\tO(d)$ as well.

\paragraph{Shift-Matrix}
We use $e_i$ to denote the $i$-th standard unit vector. We use $\mZ\in\F^{n\times n}$ to denote the \emph{shift-down matrix}, i.e., $\mZ e_i = e_{i+1}$ for all $i<n$ and $\mZ e_n = 0$. 
In particular, for any vector $v$, the product $\mZ v$ is the vector $v$ shifted down by one, 
i.e., $v = (v_1,v_2,...,v_n)^\top$ becomes $\mZ v = (0, v_1,..., v_{n-1})^\top$.

\paragraph{Frobenius Form and Generic Matrices}

Following~\cite{KarczmarzS23}, we call a matrix $\mA\in\F^{n\times n}$ \emph{generic}\footnote{In the literature, such matrices are also called cyclic~\cite{neumann1995cyclic}, single-invariant-factor, or non-derogatory~\cite{horn2012matrix}.} if there exists a vector $v\in\F^n$ such that $v,\mA v,\mA^2 v,...,\mA^{n-1} v$ are linearly independent, i.e., form a basis.
In particular, the matrix $\mK = [v|\mA v|...|\mA^{n-1} v]\in\F^{n\times n}$ (where we stack these vectors next to each other) is invertible.
Using these concepts, we define the \emph{Frobenius form} as follows. (We give more discussion in the overview \cref{sec:overview:frobenius} to build further intuition.)

\begin{definition}[Frobenius form of a generic matrix]\label{def:frobenius}
    For a generic
    matrix $\mA$, let $c\in\F^n$ be the negated coefficients of the characteristic polynomial, i.e., $x^n - \sum_{i=0}^{n-1} x^i\cdot c_{i+1} = \det(x\mI-\mA)$. Then for $\mC = \mZ + c e_n^\top$ and invertible $\mK\in\F^{n\times n}$ we call $\mA = \mK \mC \mK^{-1}$ the \emph{Frobenius form} and $\mC$ the \emph{companion matrix}. The basis-transform $\mK$ is not necessarily unique. A valid choice is ${\mK = [v|\mA v|...|\mA^{n-1} v]}$ for any $v\in\F^n$ such that $\mK$ is invertible. We call such $v$ \emph{cyclic}.
\end{definition}

The Frobenius form also exists for non-generic matrices (see, e.g., textbooks \cite{gantmakher2000theory,friedberg2018linear}), but then the structure of matrix $\mC$ is more complicated and consists of multiple blocks. Matrix $\mK$ also needs to be constructed differently.
The simplified structure for generic matrices makes algorithms simpler and more efficient.
Further, for the application of dynamic distances, the matrices we study in this work can be assumed generic (see \Cref{thm:weighted-adjacency-matrix}). That is why we use the simplified definition of Frobenius form, restricted to generic matrices.

Finding an appropriate $v$ to construct $\mK$ is simple, as we can just pick a random vector:

\begin{lemma}\label{lem:sample_cyclic_vector}
    Let $\mA$ be a generic matrix and $v\in\F^n$ be a uniformly random vector. Then with probability at least $1-n/|\F|$, matrix $\mK:=[v|\mA v|\mA^2 v|...|\mA^{n-1}v] \in \F^{n\times n}$ is invertible, i.e., $v$ is cyclic.
\end{lemma}

\begin{proof}
    The polynomial $f(v) = \det([v|\mA v|...|\mA^{n-1}v])$ is of degree $n$. Matrix $\mA$ is generic, so there exists $v$ with $f(v)\neq 0$, i.e., $f$ is a non-zero polynomial. By the Schwartz-Zippel lemma~\cite{Schwartz80,Zippel79} we thus have for random $v$ that $\P[f(v)\neq 0] \ge 1 - n/|\F|$.
\end{proof}

Since in all our algorithms (including the graph applications) we use a field $\F$ of size $\poly(n)$, the matrix $\mK$ is invertible with high probability $\ge 1-1/\poly(n)$.

\section{Technical Overview}

In this section we outline our main dynamic distance result \Cref{thm:dynamic-distances}.
The main difficulty in understanding the algorithm is that it builds heavily on machinery from computer algebra and structured linear algebra, which readers from the dynamic graph community may not be familiar with. As such, a large part of our overview focuses on explaining these algebraic tools.
\Cref{sec:overview:frobenius} goes further into the Frobenius form and the difficulties for maintaining it dynamically.
\Cref{sec:overview:toeplitzlike} explains the concept of Toeplitz-like matrices.
Equipped with these algebraic tools, \Cref{sec:overview:dynamicfrobenius} shows how to use them to maintain the Frobenius form (and thus distances) in a dynamic setting.

\subsection{Main Technical Barriers and Ideas}
\label{sec:overview:frobenius}
\paragraph{Distances via Matrix Powers:}
Given the adjacency matrix $\mA\in\Z^{n\times n}$ of a directed graph $G=(V,E)$, $(\mA^\ell)_{s,t}$ is the number of walks from $s$ to $t$ using $\ell$ steps, i.e., $(\mA^\ell)_{s,t}\neq 0$ if and only if there exists an $st$-walk with $\ell$ steps. Thus, the smallest $\ell$ such that $(\mA^\ell)_{s,t}\neq 0$ is precisely the $st$-distance in graph $G$.
Maintaining distances in graph $G$ thus reduces to the data structure task where we update matrix $\mA$ (an edge update changes one entry in adjacency matrix $\mA$) and support queries where for any given $s,t$, we return all $\mA_{s,t},(\mA^2)_{s,t},...,(\mA^{n-1})_{s,t}$.

In general, $\mA$ may not be generic and $\mA^\ell$ may contain $\tO(\ell)$-bit numbers, but via randomization (see \Cref{thm:weighted-adjacency-matrix}) we can perform computation in a finite field $\Z_p$, i.e., compute everything modulo some prime $p\le\poly(n)$, and still focus (w.h.p.) on finding the smallest $\ell$ with $(\mA^\ell)_{s,t}\neq0$.
Hence, for simplicity of this overview, one may think of $\mA$ as the classical adjacency matrix.

\paragraph{Matrix Powers via Frobenius Form:}
A common approach for computing powers is matrix decomposition, e.g., writing $\mA = \mK\mC\mK^{-1}$ for some basis transform $\mK$ such that matrix $\mC$ allows for efficient computation of $\mC^\ell$ and thus $\mA^\ell = \mK\mC\mK^{-1}\mK\mC\mK^{-1} ... =\mK \mC^\ell \mK^{-1}$.
A well-known example would be diagonalization (e.g., eigendecomposition) because taking powers of a diagonal matrix is trivial.
However, maintaining this particular decomposition is difficult since even just maintaining the eigenvalues or eigenvectors is not known in subquadratic time (see \cite{AdilS25} for a discussion of the state-of-the-art). 
Instead, we make use of the \emph{Frobenius form} (\Cref{def:frobenius}) where for a generic $\mA$, matrix $\mC$ is of the form $\mZ + c e_n^\top$, i.e., $\mC$ has 1s directly below the diagonal (due to $\mZ$), and one dense column on the right (due to $c e_n^\top$).
\[
\mA = \mK
\underbrace{\begin{bmatrix}
0&&& c_1\\
1&\ddots&& c_2\vspace{-4pt}\\
&\ddots&0&\vdots\\
&&1&c_n
\end{bmatrix}}_{=\mC}
\mK^{-1}
\]
This almost diagonal shape of $\mC$ makes it easy to compute powers $\mC^\ell$ (proven in \Cref{sec:frobenius}):
\begin{restatable}{lemma}{lemcompanionpowers}\label{lem:companion_powers}
    Given $c,u,v \in \F^n$, let $\mC = \mZ + c e_n^\top \in \F^{n\times n}$ where $\mZ$ is the shift-down matrix.
    We can compute $u^\top \mC^\ell v$ for all $\ell=0,...,n$ in $\tO(n)$ operations.
\end{restatable}
We have $(\mA^\ell)_{s,t} = u^\top \mC^\ell v$ when picking $u^\top := (e_s^\top \mK)$ as the $s$-th row of $\mK$ and $v := \mK^{-1} e_t$ as the $t$-th column of $\mK^{-1}$.
This allows us to compute $(\mA^\ell)_{s,t}$ for all $\ell=1,...,n$ via \Cref{lem:companion_powers}. 
Recall that the smallest $\ell$ where $(\mA^\ell)_{s,t}\neq 0$ is the sought $st$-distance in our graph. So given the Frobenius form $\mA=\mK\mC\mK^{-1}$, we can query the distance between any vertex pair in $\tO(n)$ time.

\paragraph{Bottleneck for Dynamic Setting:}
For the purpose of dynamic distances under edge updates, the remaining task is to maintain this Frobenius form while supporting entry updates to matrix~$\mA$.
Here the bottleneck is computing basis-transforms $\mK$ and $\mK^{-1}$, which are dense $n\times n$ matrices. As stated in \Cref{def:frobenius}, $\mK=[v|\mA v|...|\mA^{n-1}v]$ is given by repeated matrix vector products with~$\mA$.
In the dynamic setting, when an entry of $\mA$ changes, then in general $\mA^\ell v$ changes in all $O(n)$ entries for all $\ell \ge 2$. In particular, matrix $\mK$ is also densely changing in $O(n^2)$ entries.
Maintaining this new $\mK$ is the main bottleneck of prior work \cite{KarczmarzS23,FrandsenS11} and is the reason why dynamically maintaining the Frobenius form ran only in quadratic time.

\paragraph{Idea for new implicit maintenance:}
Consider matrix $\mA'$ which is $\mA$ after a sequence of $k$ entry updates (e.g., $k$ edge updates).
Assume we have the Frobenius form of the original matrix $\mA= \mK\mC\mK^{-1}$ so we now need to compute/maintain the new Frobenius form $\mA' = \mK' \mC' \mK'^{-1}$.
The main bottleneck is the computation of the new $\mK'$.
There is a matrix $\mT$ representing the transformation $\mK'=\mK\mT$. In general, this $\mT$ is a dense and full-rank matrix, so working with it directly takes $\Omega(n^2)$ time.
However, we show that $\mT$ has a special structure and is a so-called ``\emph{Toeplitz-like matrix}'' (details in the next subsection).
Such Toeplitz-like matrices of size $n\times n$ can be represented in just $O(nk)$ space using so-called ``\emph{generators}'' which are $n\times k$ matrices $\mU,\mV$ for $k\ll n$.
So rather than updating the Frobenius form explicitly, we maintain transform $\mT$ implicitly by maintaining its generators $\mU,\mV$. This then gives query access to $\mK'=\mK\mT$.
Likewise, we have access to $\mK'^{-1}=\mT^{-1}\mK^{-1}$ by maintaining generators of $\mT^{-1}$.
Before we can explain this in detail, we must first define Toeplitz-like matrices and their generators.

\subsection{Toeplitz-like Matrices, Generators, and Displacement Rank (\Cref{sec:toeplitz-like})}
\label{sec:overview:toeplitzlike}

Before defining Toeplitz-like matrices, let us first recap the structurally simpler Toeplitz matrices.
For Toeplitz matrices, the entries are constant in the diagonal direction~\cite{golub2013matrix}.
For example,
\[
\begin{bmatrix}
    a_{1} & a_{2} & a_3 & a_4 \\
          & a_{1} & a_2 & a_{3} \\
          & & a_1 & a_2 \\
          &&&a_1
\end{bmatrix}
\text{ or }
\begin{bmatrix}
    a_{4} & a_{5} & a_6 & a_{7} \\
    a_{3}      & a_{4} & a_5 & a_{6} \\
    a_2      & a_3&a_4 & a_5\\
    a_1      & a_2 & a_3 &a_4
\end{bmatrix}.
\]
Since every row is a shift of the preceding row, a matrix vector product corresponds to a convolution and can be performed in $\tO(n)$ time via fast Fourier transform~\cite{golub2013matrix,CK91}.
Also, rather than writing down all $n^2$ entries of the matrix, one can represent these matrices via an $O(n)$ dimensional vector.
Thus, Toeplitz matrices can be handled much more efficiently than general matrices.

To generalize Toeplitz matrices, in structured linear algebra and computer algebra \linebreak (e.g.,~\cite{kailath1979displacement}), one defines
the displacement operator $\nabla(\cdot)$:
\[\nabla(\mM) := \mM - \hspace{-75pt}\underbrace{~~\mZ\mM\mZ^\top}_{\text{Matrix $\mM$, shifted one position down and to the right.}}\]
where $\mZ$ shifts the matrix down, and $\mZ^\top$ shifts it to the right, i.e., $\nabla(\mM)$ takes matrix $\mM$, shifts it down and right, and then subtracts it from itself.
Since a Toeplitz matrix $\mM$ is the same row repeatedly shifted, the matrix $\nabla(\mM)$ is mostly 0 and of rank at most 2; see \Cref{fig:outline:displacementrank} for an example.

\begin{figure}[t]
    \begin{align*}
        \small
        \nabla\left(\begin{bmatrix}
            a_{1} & a_{2} & ... & a_n \\
                  & a_{1} & ... & a_{n-1} \\
                  & &\ddots & \vdots \\
                  &&&a_1
        \end{bmatrix}\right)
        =&~
        \small
        \begin{bmatrix}
            a_{1} & a_{2} & ... & a_n \\
                  & a_{1} & ... & a_{n-1} \\
                  & &\ddots & \vdots \\
                  &&&a_1
        \end{bmatrix}
        -
        \begin{bmatrix}
            0 & ... & ... & 0\\
                  & a_{1} & ... & a_{n-1} \\
                  & &\ddots & \vdots \\
                  &&&a_1
        \end{bmatrix}
        =
        \begin{bmatrix}
            a_{1} & a_{2} & ... & a_n \\
            0      & ... & ... & 0 \\
            \vdots      & & & \vdots \\
            0      &...&...&0
        \end{bmatrix}
        =
        \begin{bmatrix}1\\0\\\vdots\\0\end{bmatrix} \begin{bmatrix}a_1\\a_2\\\vdots\\a_n\end{bmatrix}^\top
    \end{align*}
    \vspace{-20pt}
    \caption{Upper triangular Toeplitz matrices have displacement rank 1. In general, Toeplitz matrices have displacement rank at most 2.}
    \label{fig:outline:displacementrank}
\end{figure}

\paragraph{Toeplitz-like matrices:}
A \emph{Toeplitz-like} matrix is a matrix where $\nabla(\mM)$ has low rank, i.e., we can write $\nabla(\mM) = \mU\mV^\top$ for some $n\times k$ matrices $\mU,\mV$ for $k \ll n$. We call such $\mU,\mV$ generators of~$\mM$, and say that $\mM$ has \emph{displacement rank} $k$ when $\nabla(\mM)$ has rank $k$.

\begin{definition}[Toeplitz-like and displacement rank~\cite{kailath1979displacement}]\label{def:toeplitzlike}
Matrix $\mM \in \F^{n \times n}$ is said to be \emph{Toeplitz-like} of displacement rank $k$ if
$
\rank\big(\nabla(\mM)\big) = k.
$
We call any pair $\mU,\mV$ with $\mU\mV^\top = \nabla(\mM)$ \emph{generators of $\mM$.} When $\mU,\mV$ are $n\times k$, then they are a \emph{minimal} generator pair.
\end{definition}

Due to the higher rank of $\nabla(\mM)$, Toeplitz-like matrices can be structurally a lot more complicated than Toeplitz matrices. When $k=n$, any matrix can be represented this way.
However, as long as a matrix $\mM$ has a lower displacement rank $k < n$, we can represent $\mM$ via generators $\mU,\mV\in\F^{n\times k}$ with $O(nk)$ numbers rather than storing the dense $O(n^2)$ entries of $\mM$.

\paragraph{Computation with Toeplitz-like matrices:}

For efficiency, we want to never write down any Toeplitz-like matrix, and instead, we work with its generators exclusively. 
First, given a Toeplitz-like matrix $\mM$ via generators, we want to construct generators of $\mM^{-1}$ without constructing $\mM$ or $\mM^{-1}$ explicitly. \cite{BostanJS08} gave such a procedure for the purpose of solving structured linear systems.

\begin{lemma}[{\cite[Lemma 14 \& Proposition 5]{BostanJS08}}]\label{lem:generator_inverse}
    Let $|\F|\ge n^\gamma$.
    Given generators $\mU,\mV \in \F^{n\times k}$ of an invertible Toeplitz-like $\mM$, there is a randomized algorithm that with probability $\ge 1- 1/n^{\gamma-2}$ 
    computes generators $\mU',\mV'\in \F^{n\times (k+2)}$ of $\mM^{-1}$ in $\tilde O(k^{\omega-1}n)$ operations. 
\end{lemma}

Next, in addition to inversion, we must also be able to multiply Toeplitz-like matrices efficiently.
In~\Cref{sec:toeplitz-like} we show the following algorithm, which makes use of fast matrix multiplication.

\begin{restatable*}[Product of Toeplitz-like matrices]{theorem}{thmtoeplitzproduct}\label{thm:toplitz-topelitz-product}
Given generators of matrices $\mA,\mB \in \F^{n\times n}$ of displacement ranks $k_a,k_b$, respectively, the product $\mA\mB$ has displacement rank at most $k_a+k_b+1$ and we can compute the generators of $\mA\mB$ in $\tilde O(\MM(k_a,k_b,k_b)n/k_b + \MM(k_b,k_a,k_a)n/k_a)$ operations.

In particular, (i) when $k_a,k_b \le k$, then this is $\tilde O(k^{\omega-1}n)$ operations; (ii) when using combinatorial matrix multiplication, then this is $\tilde O(k_ak_bn)$ operations.
\end{restatable*}

One can show that an $n\times m$ matrix also has displacement rank at most $O(m)$. So the above \Cref{thm:toplitz-topelitz-product} extends to products with arbitrary matrices, e.g., an $n\times n$ displacement rank-$k$ matrix can be multiplied with a general $n\times m$ matrix in $\tO(\MM(k,m,m)n/m + \MM(m,k,k)n/k)$ operations.
In particular, multiplying a Toeplitz-like $n\times n$ matrix with a vector (i.e., $m=1$) takes $\tilde O(kn)$ time, and multiplying it with a general $n\times n$ matrix takes $\tilde O(\MM(k,k,n)n/k)$ time.

\begin{corollary}\label{cor:overview:product}
    Given generators of matrix $\mA \in \F^{n\times n}$ with displacement rank $k$, we can compute $\mA\mB$ and $\mB\mA$ for any $\mB\in\F^{n\times n}$ in $\tilde O(\MM(k,k,n)n/k)$ operations. We can compute $\mA v$ and $v^\top \mA$ for any $v \in \F^n$ in $\tilde O(kn)$ operations.
\end{corollary}

\subsection{Dynamic Frobenius Form (\Cref{sec:frobenius,sec:dynamic})}
\label{sec:overview:dynamicfrobenius}

Now equipped with tools for Toeplitz-like matrices, let us see how to use them for maintaining the Frobenius form.
For simplicity, we will start by giving an algorithm that maintains the Frobenius form $\mA = \mK \mC \mK^{-1}$ explicitly in $\tilde O(n^2)$ time per entry update to $\mA$.
Afterward, we explain how to modify it so it runs in subquadratic time per update.

\subsubsection{Warm-up: Quadratic Time Updates}
\label{sec:overview:quadratic}
Let $\mA = \mK\mC\mK^{-1}$ be the Frobenius form of the original matrix $\mA$ computed during initialization.

Assume we perform an entry update $\mA+e_i e_j^\top$. Note that $e_i e_j^\top$ is all zero except for the entry at position $(i,j)$, so this update adds $1$ to $\mA_{i,j}$, which corresponds to an edge $(i,j)$ insertion to the adjacency matrix. Deletions are handled the same way via subtraction.
Then we have
\[
\mA+e_i e_j^\top = \mK\mC\mK^{-1} + e_i e_j^\top = \mK(\mC+\mK^{-1}e_i e_j^\top \mK)\mK^{-1} = \mK(\mC+uv^\top)\mK^{-1}.
\]
Here we defined $u := \mK^{-1}e_i$ to be the $i$-th column of $\mK^{-1}$ and $v^\top := e_j^\top \mK$ as the $j$-th row of $\mK$.
Since we know the original Frobenius form of $\mA$, we have $\mK,\mK^{-1}$ explicitly stored in memory and thus we can just read these vectors $u,v$ in $O(n)$ time.

Next, let us compute the Frobenius form of $\mC+uv^\top$, i.e., $\mC+uv^\top = \mK'\mC'\mK'^{-1}$. Then we have
\[
\mA+e_i e_j^\top = \mK(\mC+uv^\top)\mK^{-1} = \mK \mK' \mC' \mK'^{-1}\mK^{-1}. 
\]
In particular, computing and assigning $\mK \gets \mK\mK'$, $\mK^{-1}\gets\mK'^{-1}\mK^{-1}$ and $\mC \gets \mC'$ gives the new Frobenius form of $\mA+e_i e_j^\top = \mK \mC \mK^{-1}$.
Here, the computational bottlenecks are (i) computation of the Frobenius form $\mC+uv^\top = \mK'\mC'\mK'^{-1}$, and (ii) computation of the products $\mK \mK'$ and $\mK'^{-1}\mK^{-1}$.

Part (i) is given by \Cref{lem:companion_update} below. That lemma also states that $\mK',\mK'^{-1}$ are Toeplitz-like with displacement rank $O(1)$, so we can compute (ii) in $\tilde O(n^2)$ time via \Cref{cor:overview:product}, yielding the new Frobenius form for $\mA+e_i e_j^\top$ in $\tilde O(n^2)$ time.
\begin{restatable}[Frobenius Form of Companion Matrix + rank-1]{lemma}{lemcompanionupdate}\label{lem:companion_update}
Let $\mC = \mZ + ce_n^\top$ be a companion matrix 
and let $u,v\in\F^{n}$ be such that $\tilde\mC:=\mC+uv^\top$ is generic with $z\in\F^n$ cyclic (e.g., a uniformly random vector $z$ by \Cref{lem:sample_cyclic_vector}). 
Then we can compute w.h.p.~the Frobenius form
\[
\tilde \mC = \mK' \mC' \mK'^{-1}
\]
where matrices $\mK'$ and $\mK'^{-1}$ are Toeplitz-like with displacement ranks $\le 3$ and $\le 5$, respectively. In particular, $\mK',\mK'^{-1}$ are not explicitly computed but we return the generators $\mU,\mV,\mU',\mV' \in \F^{n\times 5}$ of $\mK'$ and $\mK'^{-1}$.
The time complexity is $\tilde O(n)$.
\end{restatable}

A full proof of \Cref{lem:companion_update} is given in \Cref{sec:frobenius}. Here we only explain why matrices $\mK',\mK'^{-1}$ have small displacement rank.

\paragraph{Displacement rank of $\mK',\mK'^{-1}$:}
Recall that matrix $\mK'=[z|(\mC+uv^\top)z|...|(\mC+uv^\top)^{n-1}z]$ for some vector $z$ (\Cref{def:frobenius}).
To argue that the matrix has low displacement rank, we need to show that $\nabla(\mK')$ has low rank.
For $\ell\ge 1$, consider the $(\ell+1)$-th column of $\nabla(\mK')=\mK'-\mZ\mK'\mZ^\top$ which, due to $\mZ^\top$ shifting columns one to the right, is
\begin{align}
    \underbrace{(\mC+uv^\top)^\ell z}_{\text{$(\ell+1)$-th column of $\mK'$}} - \underbrace{\mZ (\mC+uv^\top)^{\ell-1} z}_{\text{$(\ell+1)$-th column of $\mZ\mK'\mZ^\top$}}
=&~
(\mZ+ce_n^\top+uv^\top)(\mC+uv^\top)^{\ell-1}z - \mZ (\mC+uv^\top)^{\ell-1} z\notag\\
=&~
c \cdot (e_n^\top (\mC+uv^\top)^{\ell-1}z) + u \cdot (v^\top (\mC+uv^\top)^{\ell-1}z) \label{eq:overview:generator}
\end{align}
where we used that $\mC = \mZ + c e_n^\top$.
Here, $e_n^\top (\mC+uv^\top)^{\ell-1}z$ and $v^\top (\mC+uv^\top)^{\ell-1}z$ are scalars. In particular, \eqref{eq:overview:generator} is a linear combination of vector $c$ and vector $u$.
Thus, every column of $\nabla(\mK')$ is a linear combination of 3 vectors: $c$, $u$, and $z$ (the first column is equal to $z$). Hence, the displacement rank of $\mK'$ is at most 3. The inverse $\mK'^{-1}$ has displacement rank at most 5 by \Cref{lem:generator_inverse}.

\subsubsection{Subquadratic Time Updates}

Observe that in the Frobenius update procedure we outlined above (\Cref{sec:overview:quadratic}), the only quadratic time bottleneck is due to the products $\mK\mK'$ and $\mK'^{-1}\mK^{-1}$.
These products take $\tilde O(n^2)$ time, but all other operations (e.g., computing the new companion matrix $\mC'$) take only $\tilde O(n)$ time.

So instead of computing these expensive products explicitly, we introduce an additional matrix~$\mT$ accumulating the change of basis and maintain the Frobenius form in the form
\begin{align}
\mA=(\mK_\init\mT)~\mC~(\mT^{-1}\mK_\init^{-1}) \label{eq:overview:form}
\end{align}
and $\mK_\init$ is the matrix constructed during preprocessing/initialization.
As we update matrix $\mA$, the companion matrix $\mC$ and transform $\mT$ are updated. We show that after $k$ updates to $\mA$, matrix $\mT$ has displacement rank at most $O(k)$, and we store $\mT$ only via its generators.

\paragraph{Update Procedure}
Assume we have another entry update to matrix $\mA$. Then, as before,
\[
\mA+ e_i e_j^\top 
= \mK_\init\mT (\mC+\mT^{-1}\mK^{-1}_\init e_i e_j^\top \mK_\init\mT) \mT^{-1}\mK_\init^{-1}.
\]
The products $u:=\mT^{-1}\mK^{-1}_\init e_i$ and $v^\top:=e_j^\top \mK_\init\mT$ take $\tilde O(nk)$ time, by first reading a column of $\mK^{-1}_\init$ (and row of $\mK_\init$), then computing the matrix-vector-product with $\mT^{-1}$ (and $\mT$).
The latter takes $\tilde O(nk)$ time by \Cref{cor:overview:product}, as $\mT$ and $\mT^{-1}$ have displacement rank at most $O(k)$. 

Once we have these $u,v$, we compute the Frobenius form of $\mC+uv^\top = \mK'\mC'\mK'^{-1}$ as before in $\tO(n)$ time (\Cref{lem:companion_update}).
At last, we maintain the invariant (i.e., Frobenius form) as in \eqref{eq:overview:form}
\[
\mA + e_i e_j^\top = \mK_\init \mT \mK' \mC' \mK'^{-1} \mT^{-1}\mK_\init^{-1}
=
(\mK_\init \mT') \mC' (\mT'^{-1} \mK_\init^{-1}),
\]
where generators of the matrices $\mT'= \mT\mK'$ and $\mT'^{-1}=\mK'^{-1}\mT^{-1}$ can both be computed in $\tilde O(nk)$ operations via \Cref{thm:toplitz-topelitz-product}, because $\mT,\mT^{-1}$ are of displacement rank $O(k)$ and $\mK',\mK'^{-1}$ are of displacement rank $O(1)$.
The displacement rank of $\mT'$ grows by $O(1)$ by \Cref{thm:toplitz-topelitz-product}.

\paragraph{Reset}
As the displacement rank of $\mT$ grows by $O(1)$ with each update to $\mA$, we must eventually reset our data structure after $k$ updates, so that the displacement rank of $\mT$ is bounded by $O(k)$.
Instead of restarting from scratch to get a new $\mK_\init$, we compute $\mK_\init \gets \mK_\init \mT$ and \linebreak $\mK_\init^{-1} \gets \mT^{-1}\mK_\init^{-1}$ and reset $\mT$ to an identity matrix. These products take $O(\MM(k,k,n)n/k)$ time by \Cref{cor:overview:product}.
As this reset happens every $k$ updates, we obtain amortized update time
\[
\tilde O(kn + \MM(k,k,n)n/k^2).
\]
This can be made worst-case via standard techniques (e.g., \cite{AbrahamCK17,BrandNS19}).

\paragraph{Query}
Let us quickly state the query complexity of our dynamic Frobenius algorithm.
Via the maintained Frobenius form, we can answer queries that for any $s,t$ return $(\mA^\ell)_{s,t}$ for all ${\ell=1,...,n-1}$. (This yields the $st$-distance for adjacency matrix $\mA$.)
We need to compute
\[
(\mA^\ell)_{s,t} = e_s^\top (\mK_\init \mT \mC^\ell \mT^{-1} \mK_\init^{-1}) e_t
=
v^\top \mC^\ell u
\]
for $v^\top := e_s^\top\mK_\init \mT$ and $u := \mT^{-1} \mK_\init^{-1} e_t$. Here $u,v$ are computable in $\tilde O(nk)$ time as we just read a row/column of $\mK_\init,\mK_\init^{-1}$, and then multiply the corresponding vector with a displacement rank $O(k)$ matrix $\mT$ or $\mT^{-1}$, respectively (\Cref{cor:overview:product}).
The values $v^\top \mC^\ell u$ for all $\ell$ can be computed in $\tO(n)$ time via \Cref{lem:companion_powers}. So overall, processing a query takes $\tilde O(nk)$ time.

\section{Toeplitz-like Matrices}
\label{sec:toeplitz-like}

The main result of this section is the complexity of multiplying two Toeplitz-like matrices.
The concept of Toeplitz-like matrices and generators has been defined in \Cref{def:toeplitzlike}.

\thmtoeplitzproduct

A key ingredient for proving \Cref{thm:toplitz-topelitz-product} is the product of a Toeplitz-like and a general matrix.

\begin{theorem}\label{thm:toeplitz-general-matrix-product}
    Given any matrix $\mB \in \F^{n\times m}$ and the generators of a displacement rank-$k$ Toeplitz-like matrix $\mA \in \F^{n\times n}$, we can compute $\mA\mB$ and $\mB^\top \mA$ in $\tilde O(\MM(k,k,m)n/k + \MM(k,m,m)n/m) \le \tilde O(k^{\omega-2}mn+m^{\omega-2}kn)$.
\end{theorem}

An algorithm for a product as in \Cref{thm:toeplitz-general-matrix-product} has been previously given by \cite{BostanJS08}, who showed a displacement rank-$k$ Toeplitz-like matrix can be multiplied by an $n\times k$ matrix in $\tO(k^{\omega-1}n)$ operations.
Via blocking, this would imply a $\tO(k^{\omega-2}nm)$ bound (for $m\ge k$) when multiplying a displacement rank-$k$ matrix by an $n\times m$ matrix.
Since the exponent of rectangular matrix multiplication is convex, their algorithm is slower than \Cref{thm:toeplitz-general-matrix-product} for $m \gg k$.
We obtain \Cref{thm:toeplitz-general-matrix-product} by reanalyzing their algorithm using rectangular matrix multiplication and also analyze the case $k < m$.

\subsection{Toeplitz-like Matrix Products}

Toeplitz matrices can be multiplied efficiently using fast Fourier transform, i.e., polynomial multiplication (see, e.g., \cite{golub2013matrix}). 
We state the following two lemmas because we require the precise details about the indexing.

\begin{lemma}\label{lem:upper_toeplitz}
    Let $\mU$ be an upper triangular Toeplitz-matrix for vector $u\in\F^n$, i.e.,
    \[
    \mU = \begin{bmatrix}
        u_1 & u_2 & ... & u_n \\
        & u_1 & ... & u_{n-1}\\
        &&\ddots\\
        &&&u_1
    \end{bmatrix}.
    \]
    Let $v\in\F^n$ be some vector.
    Define polynomials $u(x) = \sum_{i=1}^{n} x^{i-1} u_{i}$ and $v(x) = \sum_{i=1}^{n} x^{n-i} v_{i}$.
    Then \[u(x)\cdot v(x) \mod x^n = \sum_{i=1}^{n} x^{n-i} (\mU v)_{i}.\]
\end{lemma}
We remark that polynomial $v(x)$ is defined in reverse order compared to vector $v$, and that the output polynomial is also in reverse order of vector $\mU v$.
\begin{proof}
    We have $u(x)v(x) = \sum_{a,l} x^{a-1+n-l} u_{a} v_{l}$. Thus, the coefficient of $x^{n-i}$ collects all terms with $a-1+n-l=n-i$, i.e., it equals
    $
    \sum_{l=i}^n u_{l-i+1} v_{l} = (\mU v)_i
    $.
\end{proof}

\begin{lemma}\label{lem:lower_toeplitz}
Let $\mU$ be a lower triangular Toeplitz-matrix for vector $u\in\F^n$, i.e.,
\[
\mU = \begin{bmatrix}
        u_1 &  \\
        u_2 & u_1 & \\
        \vdots &&\ddots\\
        u_n & u_{n-1}& ... &u_1
    \end{bmatrix}.
\]
Let $v\in\F^n$ be some vector.
Define $ u(x)=\sum_{i=1}^{n} x^{n-i} u_{i} $,
$ v(x)=\sum_{i=1}^{n} x^{n-i} v_{i}$. Then, for $1\le i\le n$,
\[
  (\mU v)_i \;=\; \left[x^{\,2n-1-i}\right]( u(x) v(x)),\]
where $\left[x^{2n-1-i}\right](\cdot)$ means to read the coefficient of $x^{2n-1-i}$.
\end{lemma}
We remark that polynomials $u(x),v(x)$ are defined in reverse order compared to vectors $u,v$, and that the output polynomial is also in reverse order of vector $\mU v$.
\begin{proof}
$u(x) v(x)=\sum_{a,l}u_{a} v_{l}\,x^{2n-a-l}$. So the coefficient of
$x^{2n-i-1}$ collects all terms with $a+l=i+1$, i.e.,  it equals $\sum_{l=1}^{i} u_{i+1-l}v_{l}=(\mU v)_i$.
\end{proof}

Just like Toeplitz matrices, Toeplitz-like matrices can also be multiplied efficiently via fast Fourier transform.
The following \Cref{lem:toeplitz-via-polynomial} states that connection formally.  
It shows how, when we are given the generators $\mU,\mV$ of some Toeplitz-like matrix $\mA$, the product $\mA\mB$ can be computed via polynomial multiplication.

\begin{lemma}\label{lem:toeplitz-via-polynomial}
Let $\mB\in\F^{n\times m}$, and let $\mU,\mV\in\F^{n\times k}$ be generators of a
Toeplitz-like matrix $\mathbf{A}\in\F^{n\times n}$, i.e.,
$\mathbf{A}-\mathbf{Z}\mathbf{A}\mathbf{Z}^\top=\mU\mV^\top$, with $\mathbf{Z}$ the down-shift matrix.
Define $u(x)\in\F[x]^k$, $v(x)\in\F[x]^k$, $b(x)\in\F[x]^m$ entrywise by
\[
  u_{i}(x)=\sum_{l=1}^{n}U_{l,i}\,x^{n-l},\qquad
  v_{i}(x)=\sum_{l=1}^{n}V_{l,i}\,x^{l-1},\qquad
  b_{j}(x)=\sum_{l=1}^{n}B_{l,j}\,x^{n-l},
\]
i.e., $u_i(x),v_i(x),b_j(x)$ stem from interpreting the $i$th columns of $\mU,\mV$ and the $j$th column of $\mB$ as polynomials, though the coefficient order is reversed for $u$ and $b$.
Then
\[
  (\mathbf{A}\mB)_{i,j}
  \;=\;\Big[x^{\,2n-1-i}\Big]\;
  \Big(u(x)^\top\big(v(x)\,b(x)^\top \bmod x^n\big)\Big)_j
  \qquad (1\le i\le n,\ 1\le j\le m),
\]
where $\left[x^{2n-1-i}\right](\cdot)$ means to read the coefficient of $x^{2n-1-i}$, and the $\mod x^n$ is applied entrywise to the $k\times m$ polynomial matrix $v(x)b(x)^\top$.
\end{lemma}

\begin{proof}
    We have $\nabla(\mA) = \mA - \mZ \mA \mZ^\top = \mU\mV^\top$ and thus $\mA = \mU\mV^\top + \mZ \mA \mZ^\top$. Repeatedly substituting~$\mA$ results in
    \begin{align*}
    \mA 
    =&~ 
    \mU\mV^\top + \mZ \mA \mZ^\top
    =
    \mU\mV^\top + \mZ \mU\mV^\top \mZ^\top + \mZ^2 \mA (\mZ^\top)^2
    =
    ...
    =
    \left(\sum_{\ell=0}^{n-1} \mZ^\ell \mU \mV^\top (\mZ^\top)^\ell\right) +\mZ^n \mA (\mZ^\top)^n \\
    =&~
    \sum_{\ell=0}^{n-1} \mZ^\ell \mU \mV^\top (\mZ^\top)^\ell
    \end{align*}
    because $\mZ^n = 0$.
    Interpreting $\mU\mV^\top = \sum_{i=1}^k \mU e_i (\mV e_i)^\top$ as a sum of $k$ outer products, we can write
    \begin{align*}
    \mA
    =&~
    \sum_{\ell=0}^{n-1} \mZ^\ell \left(\sum_{i=1}^k \mU e_i ( \mV e_i)^\top\right) (\mZ^\top)^\ell
    =
    \sum_{i=1}^k \sum_{\ell=0}^{n-1} (\mZ^\ell \mU e_i) (\mZ^\ell \mV e_i)^\top \\
    =&~
    \sum_{i=1}^k
    \begin{bmatrix}
        \mU_{1,i} \\
        \mU_{2,i} & \mU_{1,i} \\
        \vdots&&\ddots\\
        \mU_{n,i} & \mU_{n-1,i} & ... & \mU_{1,i}
    \end{bmatrix}
    \begin{bmatrix}
        \mV_{1,i} & \mV_{2,i} & ... &\mV_{n,i} \\
        & \mV_{1,i} & ... &\mV_{n-1,i}\\
        &&\ddots\\
        &&&\mV_{1,i}
    \end{bmatrix}.
    \end{align*}
    The last step uses that $\mZ^\ell$ shifts vectors down by $\ell$.
    Now consider multiplying vector $\mB e_j$ from the right, which requires computing for each $i=1,...,k$ the following matrix vector products
    \[
    y :=
    \begin{bmatrix}
        \mV_{1,i} & \mV_{2,i} & ... &\mV_{n,i} \\
        & \mV_{1,i} & ... &\mV_{n-1,i}\\
        &&\ddots\\
        &&&\mV_{1,i}
    \end{bmatrix}
    \mB e_j
    ~\text{ and }~
    z:=\begin{bmatrix}
        \mU_{1,i} \\
        \mU_{2,i} & \mU_{1,i} \\
        \vdots&&\ddots\\
        \mU_{n,i} & \mU_{n-1,i} & ... & \mU_{1,i}
    \end{bmatrix} y.
    \]
    These products $y$ and $z$ can be performed via \Cref{lem:upper_toeplitz} (upper triangular) and \Cref{lem:lower_toeplitz} (lower triangular), respectively.
    Since polynomial $b_j(x)$ represents column $\mB e_j$ in reversed order, by \Cref{lem:upper_toeplitz} the coefficients of $x^0,\ldots,x^{n-1}$ in the product $v_i(x)b_j(x)\mod x^n$ correspond precisely to the vector $y$ in reverse order.
    Since we defined polynomial $u_i(x)$ in reverse order and vector $y$ is given as a polynomial in reverse order, this matches exactly the orders needed to compute $z$ via \Cref{lem:lower_toeplitz}.
    Thus the product $z$ is given by $u_i(x)(v_i(x)b_j(x)\mod x^n)$, i.e., for all $1\le \ell \le n$ we have
    \[
    \left(\begin{bmatrix}
        \mU_{1,i} \\
        \mU_{2,i} & \mU_{1,i} \\
        \vdots&&\ddots\\
        \mU_{n,i} & \mU_{n-1,i} & ... & \mU_{1,i}
    \end{bmatrix}
    \begin{bmatrix}
        \mV_{1,i} & \mV_{2,i} & ... &\mV_{n,i} \\
        & \mV_{1,i} & ... &\mV_{n-1,i}\\
        &&\ddots\\
        &&&\mV_{1,i}
    \end{bmatrix}
    \mB e_j\right)_\ell
    =
    [x^{2n-1-\ell}](u_i(x)(v_i(x)b_j(x)\mod x^n)).
    \]
    Computing the $j$-th column of $\mA\mB$ thus reduces to computing $\sum_{i=1}^k u_i(x)(v_i(x)b_j(x) \mod x^n) = u(x)^\top (v(x)b_j(x) \mod x^n)$.
    Therefore, the full product $\mA\mB$ reduces to computing\linebreak $u(x)^\top (v(x)b(x)^\top \mod x^n)$.
\end{proof}

The main difficulty in computing $u(x)^\top (v(x)b(x)^\top \mod x^n)$ in~\Cref{lem:toeplitz-via-polynomial} is that the modulo operation enforces multiplication from right to left.
To see why that is an issue, consider the case where there is no modulo. Then we could use associativity and compute $(u(x)^\top v(x)) b(x)^\top$ from left to right.

\emph{Left-to-right (if there were no modulo):}
The first product $u(x)^\top v(x)$ is a dot-product and computable in $\tilde O(kn)$ time because $u(x),v(x)\in\F[x]^k$ have degree $n$.
Then multiplying by $b(x)^\top\in\F[x]^m$ takes another $\tilde O(mn)$ time.

\emph{Right-to-left (enforced by modulo):}
Because of the modulo, the multiplication order must be right to left: $u(x)^\top (v(x) b(x)^\top \mod x^n)$, where $v(x) b(x)^\top$ results in a $k\times m$ polynomial matrix of degree at most $2n$. Thus multiplying with $u(x)^\top$ then takes $\tilde O(kmn) \gg \tilde O(kn+km)$ time.

The following \Cref{lem:remove-modulo} allows us to remove the modulo, allowing us to use associativity, at the cost of needing several smaller polynomial products.

\begin{lemma}[{\cite[Lemma 12]{BostanJS08}}]\label{lem:remove-modulo}
    Given polynomials $p,q \in \F[x]$ with $p = p_0+p_1x+p_2x^2+...$, define $p^{(\ell,k)}$ as follows:
    \[
    p^{(\ell,1)} = p_\ell \quad \text{and} \quad p^{(\ell,k)} = \sum_{i=\ell k}^{\ell k+k/2-1}p_ix^{i-\ell k} \text{ for $k\ge 2$}
    \]
    and likewise define $q^{(\ell,k)}$.
    Then for $n$ a power of 2,
    \[
    pq \mod x^n = \sum_{k=0}^{\log n} x^{n-2^k} \sum_{\ell=0}^{n/2^k-1} p^{(\ell,2^k)}q^{(n/2^k-1-\ell,2^k)}.
    \]
\end{lemma}

Getting rid of the modulo operation allows us to use associativity, leading to the following corollary:

\begin{corollary}\label{cor:no_modulo}
    Given polynomials $u,v,b \in \F[x]$ and $n$ a power of two, we have
    \[
    u(vb\mod x^n) = \sum_{k=0}^{\log n} x^{n-2^k} \sum_{\ell=0}^{n/2^k-1} uv^{(\ell,2^k)}b^{(n/2^k-1-\ell,2^k)}.
    \]
\end{corollary}

We can now prove \Cref{thm:toeplitz-general-matrix-product}. Recall that it states we can multiply a displacement rank-$k$ matrix by an $n\times m$ matrix in $\tilde O(\MM(k,k,m)n/k + \MM(k,m,m)n/m)$ operations.

\begin{proof}[Proof of \Cref{thm:toeplitz-general-matrix-product}]
    First, observe  $\nabla(\mA^\top) = \mA^\top - \mZ\mA^\top\mZ^\top = (\mA - \mZ \mA \mZ^\top)^\top = \nabla(\mA)^\top = \mV\mU^\top$, so the generators $\mU\mV^\top$ of $\mA$ are also the generators $\mV\mU^\top$ of $\mA^\top$. In particular, by $\mB^\top \mA = (\mA^\top \mB)^\top$, the product $\mB^\top \mA$ can be done in the same time complexity as $\mA^\top \mB$, which is the same time complexity as computing $\mA \mB$. So for the remainder of the proof, let us focus on computing $\mA\mB$.

    Further, let us assume that $n$ is a power of 2, as otherwise we can simply add 0-columns and rows, which increases the complexity by at most a constant factor.
    
    By \Cref{lem:toeplitz-via-polynomial} it suffices to compute $u^\top(vb^\top\mod x^n)$ for $u,v\in\F[x]^k$ and $b\in\F[x]^m$.
    Consider computing the $j$-th column of that product.

    \begin{align*}
        (u^\top(vb^\top\mod x^n))_{j} 
        =&~
        \sum_{t=1}^k u_t (vb^\top\mod x^n)_{t,j}\\
        =&~
        \sum_{t=1}^k u_t (v_t b_j\mod x^n)\\
        \text{(\Cref{cor:no_modulo})}~~=&~
        \sum_{t=1}^k \sum_{s=0}^{\log n} x^{n-2^s} \sum_{\ell=0}^{n/2^s-1} u_t v_t^{(\ell,2^s)}b_j^{(n/2^s-1-\ell,2^s)}\\
        =&~
        \sum_{s=0}^{\log n} x^{n-2^s} \sum_{\ell=0}^{n/2^s-1} \sum_{t=1}^k u_t v_t^{(\ell,2^s)}b_j^{(n/2^s-1-\ell,2^s)}\\
        =&~
        \sum_{s=0}^{\log n} x^{n-2^s} \sum_{\ell=0}^{n/2^s-1} u^\top v^{(\ell,2^s)}b_j^{(n/2^s-1-\ell,2^s)} \\
        =&~
        \sum_{s=0}^{\log n} x^{n-2^s} u^\top \mV^{(s)} \mB^{(s)} e_j
    \end{align*}
    where we define $\mV_{t,\ell}^{(s)}=v^{(\ell,2^s)}_t$ and $\mB^{(s)}_{\ell,j} = b_j^{(n/2^s-1-\ell,2^s)}$ for $t=1,...,k$, $\ell=0,...,n/2^s-1$, $j = 1,...,m$, i.e., $\mV^{(s)} \in \F[x]^{k\times n/2^s}$ and $\mB^{(s)} \in \F[x]^{n/2^s \times m}$, both of degree $\le 2^s$.
    
    So to compute $u^\top(vb^\top\mod x^n)$, we need to compute $u^\top \mV^{(s)}\mB^{(s)}$ for all $s=0,1,2,...,\log n$.
    Here $\mV^{(s)},\mB^{(s)}$ have degrees at most $2^s$, but $u$ has degree at most $n$. To balance that, we express $u^\top$ as \[u^\top = (1, x^{2^s}, x^{2\cdot 2^s}, ..., x^{n-2^s}) \mN,\] 
    i.e., each polynomial in $u$ is split into $n/2^s$ polynomials of degree $2^s$ and then we stack them on top of each other, resulting in $\mN$.
    We must now compute
    \[
    (1, x^{2^s}, x^{2\cdot 2^s}, ..., x^{n-2^s}) \mN \mV^{(s)} \mB^{(s)}.
    \]
    There are two options for computing this product: multiplying right-to-left or left-to-right, which have different time complexities.
    \paragraph{Option 1: right-to-left}
    If we multiply from right to left, we have time complexity 
    \[\tilde O(\MM(k,n/2^s, m)\cdot 2^s + \MM(n/2^s,k, m)\cdot 2^s) = \tilde O(\MM(n/2^s,k, m)\cdot 2^s).\]
    This complexity is increasing for increasing $s$ due to $\MM(n/2^{s},k,m) \le \MM(n/2^{s+1}, k,m)\cdot 2$ (i.e., perform a matrix product by splitting it into 2 smaller products) and thus $\MM(n/2^s,k, m)\cdot 2^s \le \MM(n/2^{s+1},k, m)\cdot 2^{s+1}$.
    Note that the final product by a length-$n/2^s$ vector of monomials is essentially reverting the polynomial split, so it takes negligible $O(nm)$ time.
    
    We use option 1 for all $s$ where $n/2^s \ge \min(m,k)$.
    Since the complexity increases for increasing $s$, the complexity is bounded by the largest $s$, i.e.,
    $$
    \tilde O(\MM(\min(m,k),k,m)\cdot(n/\min(m,k))) \le \tilde O(\MM(m,k,m)\cdot n/m + \MM(k,k,m)\cdot n/k).
    $$
    \paragraph{Option 2: left-to-right}
    If we first multiply $\mN\mV^{(s)}$ and then multiply by $\mB^{(s)}$, then the time complexity is
    \[\tilde O(\MM(n/2^s, k, n/2^s)\cdot 2^s + \MM(n/2^s, n/2^s, m) \cdot 2^s).\]
    Since we can reorder the arguments, only one of the two terms above dominates, depending on whether $m$ or $k$ is larger.
    Assume for now without loss of generality that $m\ge k$ (otherwise switch $m$ and $k$ in the analysis below).
    We use option 2 only when $n/2^s < \min(m,k) = k$, so we can write $n/2^s = k/2^t$ for $0\le t \le \log k$.
    Next, we write
    $$
    \MM(n/2^s, m, n/2^s)\cdot 2^s = \MM(k/2^t, m, k/2^t)\cdot 2^t \cdot n/k = k^{\omega(1-t/\log k,~ \log_k m,~ 1-t/\log k)+t/\log k} \cdot n/k
    $$
    where $\omega(a,b,c)$ is the complexity exponent of multiplying a $k^a\times k^b$ matrix by a $k^b\times k^c$ matrix.
    It is known that this exponent is a convex function \cite{LottiR83}, so we can bound the exponent by
    \begin{align*}
    &~
    \omega\left(1-\frac{t}{\log k},~ \log_k m,~ 1-\frac{t}{\log k}\right) + \frac{t}{\log k}\\
    \le&~
    \left(\omega\left(1, \log_k m, 1\right) \cdot \left(1-\frac{t}{\log k}\right) + \omega\left(0, \log_k m, 0\right) \cdot \left(\frac{t}{\log k}\right)\right) +\frac{t}{\log k}\\
    \le&~
    \omega\left(1, \log_k m, 1\right) \cdot \left(1-\frac{t}{\log k}\right) + (1+\log_k m) \cdot \left(\frac{t}{\log k}\right)\\
    \le&~
    \omega\left(1, \log_k m, 1\right) \cdot \left(1-\frac{t}{\log k}\right) + \omega(1,\log_k m, 1) \cdot \left(\frac{t}{\log k}\right)\\
    =&~
    \omega\left(1, \log_k m, 1\right)
    \end{align*}
    where the first step uses convexity, the second step uses $\omega(a,b,c) \le a+b+c$, and step three uses that $\omega(a,b,c)\ge$ the sum of the two largest arguments.
    Thus when $k\le m$ then the time complexity of option 2 is bounded by
    \[
    \tilde O(k^{\omega(1,\log_k m, 1)} \cdot n/k)=\tilde O( \MM(k,m,k)\cdot n/k).
    \]
    For the case that $m\le k$, the complexity of option 2 is bounded by $\tilde O(\MM(m,k,m)\cdot n/m)$ (switch all $m$ and $k$ in the above analysis).
    \paragraph{Overall complexity}
    By choosing option 1 when $n/2^s \ge \min(m,k)$ and option 2 when $n/2^s < \min(m,k)$, the time complexity for each $s$ is bounded by
    \[
    \tilde O(\MM(k,k,m)n/k + \MM(k,m,m)n/m).
    \]
\end{proof}

\begin{proof}[Proof of \Cref{thm:toplitz-topelitz-product}]
    Let $\nabla(\mA) = \mA - \mZ\mA\mZ^\top = \mU_A\mV_A^\top$ and likewise $\nabla(\mB) = \mU_B\mV_B^\top$ be the generators of $\mA$ and $\mB$.
    Then we have
    \begin{align*}
    \nabla(\mA\mB) 
    = &~
    \mA\mB - \mZ\mA\mB\mZ^\top\\
    =&~
    \underbrace{(\nabla(\mA)+\mZ\mA\mZ^\top)}_{\mA}\mB - \mZ\mA\mB\mZ^\top\\
    =&~
    \nabla(\mA)\mB + \mZ\mA\mZ^\top\mB - \mZ\mA\mB\mZ^\top\\
    =&~
    \nabla(\mA)\mB + \mZ\mA\mZ^\top\underbrace{(\nabla\mB+\mZ\mB\mZ^\top)}_{\mB} - \mZ\mA\mB\mZ^\top \\
    =&~
    \nabla(\mA)\mB + \mZ\mA\mZ^\top\nabla\mB+\mZ\mA\mZ^\top\mZ\mB\mZ^\top - \mZ\mA\mB\mZ^\top \\
    =&~
    \nabla(\mA)\mB + \mZ\mA\mZ^\top\nabla\mB+\mZ\mA(\mZ^\top\mZ-\mI)\mB\mZ^\top
    \end{align*}
    Next, we use that $\mZ^\top\mZ = \mI-e_ne_n^\top$, and plug in $\nabla(\mA)=\mU_A\mV_A^\top$, $\nabla(\mB)=\mU_B\mV_B^\top$ to get
    \begin{align*}
    \nabla(\mA\mB) 
    =&~
    \nabla(\mA)\mB + \mZ\mA\mZ^\top\nabla\mB+\mZ\mA(\mZ^\top\mZ-\mI)\mB\mZ^\top \\
    =&~
    \nabla(\mA)\mB + \mZ\mA\mZ^\top\nabla\mB-\mZ\mA e_ne_n^\top\mB\mZ^\top \\
    =&~
    \mU_A\mV_A^\top\mB + \mZ\mA\mZ^\top\mU_B\mV_B^\top
    -(\mZ\mA e_n) (e_n^\top\mB\mZ^\top) \\
    =&~
    \underbrace{[\mU_A | \mZ\mA\mZ^\top \mU_B | -(\mZ\mA e_n)]}_{=:\mU_{AB}} \underbrace{[\mB^\top\mV_A | \mV_B | \mZ\mB^\top e_n]^\top}_{=:\mV_{AB}^\top}
    \end{align*}
    So we have our new generators $\mU_{AB},\mV_{AB} \in \F^{n \times (k_a+k_b+1)}$ for the matrix product $\mA\mB$.

    \paragraph{Time Complexity}
    A matrix-vector product with $\mA$ or $\mB$ takes $\tilde O(k_a n)$ and $\tilde O(k_b n)$ time, respectively, by \Cref{thm:toeplitz-general-matrix-product}, allowing us to compute $\mZ\mA e_n$ and  $\mZ\mB^\top e_n$.
    
    We can compute
    $\mZ\mA(\mZ^\top\mU_B)$ and $\mB^\top \mV_A = (\mV_A^\top \mB)^\top$ in $\tilde O(\MM(k_a,k_a,k_b)n/k_a + \MM(k_a,k_b,k_b)n/k_b)$ time using \Cref{thm:toeplitz-general-matrix-product}. This is because $\mA$ has displacement rank $k_a$ and $\mZ^\top \mU_B$ has $k_b$ columns. Likewise, $\mB$ has displacement rank $k_b$ and $\mV_A$ has $k_a$ columns.
\end{proof}

\section{Updating the Frobenius Form}
\label{sec:frobenius}

In this section, we present the algorithm for maintaining the Frobenius form under low-rank updates (\Cref{thm:lowrankfrobenius}).
For the prior dynamic algebraic primitive of matrix inverse, this was known via the Woodbury identity.
We now show an equivalent for our new primitive, the Frobenius form.
In the context of dynamic algorithms, we will use this result in \Cref{sec:dynamic} to reset our data structure after every $k$ updates.
\begin{restatable}[\Cref{thm:lowrankfrobenius}, Low-rank update to Frobenius form]{theorem}{lowrankfrobenius}\label{thm:lowrank_frobenius_update}
    Given are the Frobenius form $\mA = \mK \mC \mK^{-1}$ of $\mA$, where $\mK,\mK^{-1}$ are explicitly given, and matrices $\mU,\mV\in\F^{n\times k}$. Assuming all prefix updates $\mA+\sum_{i=1}^\ell\mU e_i e_i^\top\mV^\top$ (i.e., using only the first $\ell$ columns of $\mU,\mV$) for $\ell=1,...,k$ are generic, we can compute w.h.p.~the new Frobenius form
    \[
    (\mA+\mU\mV^\top) = \mK'\mC'\mK'^{-1}
    \]
    in $\tilde O(\MM(n,k,k)n/k)$ operations. Here $\mK',\mK'^{-1}$ are explicitly computed.
\end{restatable}

\paragraph{Proof and Section outline:}
To prove \Cref{thm:lowrank_frobenius_update}, we first reduce the problem of computing the Frobenius form of $\mA+\mU\mV^\top$ to computing the Frobenius form of $\mC+\mU'\mV'^\top$ for some $\mU',\mV' \in \F^{n\times k}$. The latter is much easier to handle.

Let $\mA = \mK \mC \mK^{-1}$
and suppose we have an update $\mA'=\mA+\mU\mV^\top$.
We have
\[
\mA+\mU\mV^\top = \mK \mC \mK^{-1} + \mU\mV^\top = \mK (\mC + \mK^{-1}\mU\mV^\top\mK)\mK^{-1}.
\]
Thus after computing $\mU' = \mK^{-1}\mU$, $\mV'^\top = \mV^\top \mK$,
we now have a rank-$k$ update to the companion matrix $\mC + \mU'\mV'^\top$.
We compute the Frobenius form 
\[\mC + \mU'\mV'^\top = \mK' \mC' \mK'^{-1}.\]
Then we have
\[
\mA+\mU\mV^\top = \mK (\mC + \mK^{-1}\mU\mV^\top\mK)\mK^{-1} = (\mK\mK') \mC' (\mK'^{-1}\mK^{-1}).
\]
So we get our new Frobenius form for $\mA+\mU\mV^\top$ by 
(i) computing $\mU',\mV'$, (ii) computing the Frobenius form of $\mC+\mU'\mV'^\top$, and (iii) multiplying $\mK\mK'$ and $\mK'^{-1}\mK^{-1}$.

Here (i) is done with rectangular matrix multiplication, and (iii) uses the fact that $\mK'$ and $\mK'^{-1}$ have low displacement rank, so we can compute the product via the algorithms from \Cref{sec:toeplitz-like}.

So the only remaining task is (ii), computing the Frobenius form of $\mC+\mU'\mV'^\top$.
We do this via a divide-and-conquer approach.
In \Cref{sec:frobenius:companion-rank1-update} we handle the base case where $\mU',\mV'$ specify a rank-1 update, i.e., $\mU',\mV'$ are a single column vector each.
In \Cref{sec:frobenius:low-rank-update} we extend the result to rank-$k$ updates via recursion. 
There, we also combine everything to prove \Cref{thm:lowrank_frobenius_update}.

\subsection{Rank-1 Update to Companion Matrix}
\label{sec:frobenius:companion-rank1-update}

We start by giving an algorithm for computing the Frobenius form of a matrix $\mC + uv^\top$ where $\mC$ is a companion matrix.

\lemcompanionupdate*

\begin{proof}
    Let us start by computing $\mC'$, the companion matrix of $\tilde \mC$. To compute it, we need the characteristic polynomial of $\tilde \mC$.
    \paragraph{Computing characteristic polynomial}
    Applying the determinant lemma to 
    ${\det((x\mI-\mC)-uv^\top)}$
    leads to
    \[
    \det(x\mI-\mC)(1-v^\top(x\mI-\mC)^{-1}u).
    \]
    Note that the characteristic polynomial $\det(x\mI-\mC)$ of $\mC$ is given by the last column of $\mC$, since it is a generic companion matrix.
    
    Now the main computational task is to compute $(1-v^\top(x\mI-\mC)^{-1}u)$.
    Here
    \[
    (x\mI-\mC)^{-1}
    =
    x^{-1}(\mI-x^{-1}\mC)^{-1}
    =
    x^{-1} \sum_{\ell \ge 0}\mC^\ell x^{-\ell}
    =
    \sum_{\ell \ge 0}\mC^\ell x^{-\ell-1}.
    \]
    Further, we know that $\det(x\mI-\tilde\mC)$ and $\det(x\mI-\mC)$ are degree $n$ polynomials, so it suffices to compute $(x\mI-\mC)^{-1}$ up to degree $-n$.
    Next,
    \[
    v^\top(x\mI-\mC)^{-1}u
    =
    \sum_{\ell \ge 0}v^\top\mC^\ell u x^{-\ell-1}
    \]
    where $v^\top\mC^\ell u$ for all $\ell=0...n$ can be computed in $\tilde O(n)$ time since $\mC$ is a companion matrix (\Cref{lem:companion_powers}).
    Thus we can compute the new characteristic polynomial $\det(x\mI-\tilde\mC)$ in $\tilde O(n)$ time.
    
    \paragraph{Computing generators of $\mK$}
    Matrix $\mK$ is given by $[z|\tilde\mC z|...|\tilde\mC^{n-1}z]$.
    
    Write $k_j=\tilde\mC^{j-1}z$ (i.e., the $j$th column of $\mK$)
    and using $\tilde\mC=\mZ+ce_n^\top+uv^\top$ we get for $j\ge 2$
    \[
    (\mK - \mZ \mK \mZ^\top ) e_j 
    =
    k_j-\mZ k_{j-1}
    =
    (\tilde \mC - \mZ) k_{j-1}
    =
    (ce_n^\top+uv^\top) k_{j-1}
    =
    c(e_n^\top k_{j-1})+u(v^\top k_{j-1}).\]
    So every column $k_j$ of $\mK - \mZ \mK \mZ^\top$ for $j\ge 2$ is a linear combination of $c$ and $u$ (and $k_1=z$ for the first column).
    Thus $\mK - \mZ \mK \mZ^\top$ is a rank $\le3$ matrix, meaning $\mK$ has displacement rank $\le3$.
    Further, generators of $\mK$ are given by $\mU = [z|c|u]$ and 
    \[
    \mV^\top = \begin{bmatrix}
        1 & 0 & 0 &... & 0 \\
        0 & e_n^\top k_{1} & e_n^\top k_{2} & ... & e_n^\top k_{n-1} \\
        0 & v^\top k_{1} & v^\top k_{2} & ... & v^\top k_{n-1}
    \end{bmatrix}
    = 
    \begin{bmatrix}
        1 & 0 & 0 &... & 0 \\
        0 & e_n^\top z & e_n^\top \tilde\mC z & ... & e_n^\top \tilde\mC^{n-2} z \\
        0 & v^\top z & v^\top \tilde\mC z & ... & v^\top \tilde\mC^{n-2} z
    \end{bmatrix}.
    \]
    Here $e_n^\top\tilde\mC^{j} z$ and $v^\top \tilde\mC^j z$ can be computed for all $j=0...n-1$ 
    by computing the polynomials 
    \[
    v^\top (\mI - x \tilde\mC)^{-1} z \mod x^n 
    \text{ and }
    e_n^\top (\mI - x \tilde\mC)^{-1} z \mod x^n,
    \] 
    since $(\mI - x \tilde\mC)^{-1} = \sum_{j\ge0} x^j\cdot \tilde\mC^j$.
    To compute this efficiently in $\tilde O(n)$ 
    operations, we use the Sherman-Morrison identity, focusing on the pair $v,z$ (the pair $e_n, z$ works analogously).
    \[
    v^\top(\mI-x\tilde\mC)^{-1} z
    =
    v^\top(\mI-x\mC-x\cdot uv^\top)^{-1} z
    =
    v^\top(\mI-x\mC)^{-1}z+\frac{x\cdot v^\top(\mI-x\mC)^{-1}uv^\top(\mI-x\mC)^{-1}z}{1-x\cdot v^\top(\mI-x\mC)^{-1}u}.
    \]
    As all computations are performed modulo $x^n$, the products with $(\mI-x\mC)^{-1}$ are computable in $\tilde O(n)$ operations via \Cref{lem:companion_powers} because $(\mI-x\mC)^{-1} = \sum_j x^j \mC^j$ and $\mC$ is a companion matrix.
    
    \paragraph{Computing generators of $\mK^{-1}$}
    Given the generators of $\mK$, we can get the generators of the inverse $\mK^{-1}$ in $\tilde O(n)$ operations via \Cref{lem:generator_inverse}. The size of the inverse generators is at most $5$, since the generators of $\mK$ have size $3$ and \Cref{lem:generator_inverse} increases it by at most $2$.
\end{proof}

The above proof for \Cref{lem:companion_update} made use of \Cref{lem:companion_powers}. We now give its proof.

\lemcompanionpowers*

The proof is based on techniques from \cite{BostanM21}.

\begin{proof}%
Observe that $(\mI-x\cdot \mC)^{-1} = \sum_{i=0}^n x^i \mC^i \mod x^{n+1}$. So to compute $u^\top \mC^k v$ for all $k=1...n$ we can compute $u^\top (\mI - x\mC)^{-1} v \mod x^{n+1}$ and read off the coefficients of that polynomial.
All remaining calculations in this proof are done modulo $x^{n+1}$, i.e., over the ring $\F[x]/\langle x^{n+1}\rangle$.

We define $y=(\mI - x\mC)^{-1} v$; thus, $v = (\mI - x\mC)y = y - x (\mZ y + c e_n^\top y)$, implying for any $i\in[n]$
\[y_i - x (y_{i-1} + c_i y_n) = v_i
~~\Longrightarrow~~
y_i = v_i + x (y_{i-1} + c_i y_n).\]
Here, we assume $y_0=0$. Below, for convenience, we also set $c_j=0$ and $v_j=0$ for $j\notin[n]$.
Unrolling the recursion yields
\[y_i= \sum_{j=1}^i x^{i-j} v_j + y_n \sum_{j=1}^i x^{i-j+1} c_j\]
which for $i=n$ gives
\[y_n = \sum_{j=1}^n x^{n-j} v_j + y_n \sum_{j=1}^n x^{n-j+1} c_j
\Longrightarrow
y_n  = \frac{\sum_{j=1}^n x^{n-j} v_j}{1 - \sum_{j=1}^n x^{n-j+1} c_j }\]
with the division being performed modulo $x^{n+1}$, allowing us to compute this $y_n$ in $\tO(n)$ time.

Plugging this $y = (\mI - x \mC)^{-1} v$ into our desired calculation yields
\begin{align*}
u^\top (\mI - x \mC)^{-1} v
=&~
u^\top y
=
\sum_{i=1}^n u_i (v_i + x (y_{i-1} + c_i y_n)) \\
=&~
u^\top v + y_n x \cdot u^\top c + x \sum_{i=1}^n u_i y_{i-1}\\
=&~
u^\top v + y_nx\cdot u^\top c 
+ x \sum_{i=1}^n u_i (v_{i-1} + x (y_{i-2} + c_{i-1} y_n))\\
=&~
\left(u^\top v + y_nx\cdot u^\top c  \right)
+ x \left(\sum_{i=1}^n u_i v_{i-1} + y_nx\sum_{i=1}^n u_i c_{i-1}\right) + x^2 \sum_{i=1}^n u_i y_{i-2}\\
=&~
\sum_{k=0}^{n} x^k \left(\sum_{i=1}^n u_i v_{i-k}\right) + y_n\sum_{k=0}^n x^{k+1} \left( \sum_{i=1}^n u_i c_{i-k}\right).
\end{align*}
Above, in the last step, we unrolled the recurrence.
Note that the inner sums are convolutions of $u,v$ and $u,c$ and can be computed for all $k=0...n$ at once in $\tO(n)$ via fast Fourier transform~\cite{CK91}.
Multiplying by $y_n$ and adding the polynomials is also bounded by $\tO(n)$ operations.
\end{proof}

\subsection{Low-Rank Update}
\label{sec:frobenius:low-rank-update}

\Cref{sec:frobenius:companion-rank1-update} considered computing the Frobenius form of $\mC + uv^\top$. We now generalize it to $\mC + \mU\mV^\top$.
The following \Cref{cor:companion_update} is proven via divide-and-conquer, where \Cref{lem:companion_update} from \Cref{sec:frobenius:companion-rank1-update} forms the base case.

\begin{restatable}{lemma}{companionupdate}\label{cor:companion_update}
    Let $\mC = \mZ + ce_n^\top$ be a companion matrix 
    and let $\mU,\mV\in\F^{n\times k}$ such that all prefix updates $\mC+\sum_{i=1}^\ell\mU e_i e_i^\top\mV^\top$ (i.e., using only the first $\ell$ columns of $\mU,\mV$) for $\ell=1,...,k$ are generic.
    Then we can compute w.h.p.~the Frobenius form
    \[
    \tilde \mC = \mK \mC' \mK^{-1}
    \]
    where matrices $\mK,\mK^{-1}$ are Toeplitz-like with displacement rank $\le 6k-1$. 
    In particular, $\mK,\mK^{-1}$ are not explicitly computed but we return the generators of $\mK$ and $\mK^{-1}$.
    The time complexity is $\tilde O(k^{\omega-1}n)$.
\end{restatable}

\begin{proof}
    We solve this via divide-and-conquer, splitting the rank-$k$ update in half. We may assume $k$ is a power of two by appending 0 columns to $\mU$ and $\mV$. This increases the time complexity by only a constant factor.
    \paragraph{Base Case:}
    When $k=1$, we simply use \Cref{lem:companion_update}. It runs in $\tilde O(n)$ time and returns generators for $\mK,\mK^{-1}$ with displacement rank $\le 5 = 6k-1$.
    \paragraph{Recursion:}
    Let $\mU=[\mU_1|\mU_2],\mV=[\mV_1|\mV_2]$, so a rank-$k$ update $\mU\mV^\top = \mU_1\mV_1^\top+\mU_2\mV_2^\top$ becomes two rank-$k/2$ updates.
    We can compute via recursion:
    \begin{align*}
    \mC + \mU \mV^\top
    =&~
    (\mC + \mU_1\mV_1^\top) + \mU_2\mV_2^\top \\
    \begin{tabular}{c}
    \text{Perform rank-$k/2$ update to}\\
    \text{compute $\mC+\mU_1\mV_1^\top=\mK_1\mC_1\mK_1^{-1}$.}
    \end{tabular}\quad=&~
    \mK_1 \mC_1 \mK_1^{-1}  + \mU_2\mV_2^\top \\
    =&~
    \mK_1 (\mC_1+\underbrace{\mK_1^{-1}\mU_2}_{=:\mU'}\underbrace{\mV_2^\top\mK_1}_{=:\mV'^\top}) \mK_1^{-1} \\
    =&~
    \mK_1 (\mC_1+\mU'\mV'^\top) \mK_1^{-1} \\
    \begin{tabular}{c}
    \text{Perform rank-$k/2$ update to}\\
    \text{compute $\mC_1+\mU'\mV'^\top=\mK_2\mC_2\mK_2^{-1}$.}
    \end{tabular}\quad=&~
    \mK_1 \mK_2 \mC_2 \mK_2^{-1} \mK_1^{-1}\\
    \begin{tabular}{c}
    \text{Compute $\mK=\mK_1\mK_2$}\\
    \text{and $\mK^{-1}=\mK_2^{-1}\mK_1^{-1}$.}
    \end{tabular}\quad=&~
    \mK \mC_2 \mK^{-1}.
    \end{align*}
    Note that the prefix genericity property is guaranteed in the recursive calls by the fact that genericity is invariant under similarity.
    Here computing $\mU',\mV'$ takes $\tilde O(k^{\omega-1}n)$ operations via \Cref{thm:toeplitz-general-matrix-product} because by the inductive hypothesis $\mK_1,\mK_1^{-1}$ have displacement rank $\le 6(k/2)-1$, and $\mU_2,\mV_2$ have $k/2$ columns.
    Computing $\mK,\mK^{-1}$ takes $\tilde O(k^{\omega-1}n)$ operations via \Cref{thm:toplitz-topelitz-product} because $\mK_1,\mK_1^{-1},\mK_2,\mK_2^{-1}$ have displacement rank at most $6(k/2)-1$.
    Further, $\mK,\mK^{-1}$ have displacement rank $\le \left(6(k/2)-1\right)+\left(6(k/2)-1\right)+1 = 6k-1$, so the claimed bound in \Cref{cor:companion_update} holds.

    \paragraph{Time Complexity}
    The time complexity $T(k)$ of this recursive algorithm is described by \linebreak $T(k) = \tilde O(k^{\omega-1}n) + 2T(k/2)$ and $T(1)=\tilde O(n)$, leading to $T(k)=\tilde O(k^{\omega-1}n)$ as ${\omega - 1 \ge 1 = \log_2(2)}$.
\end{proof}

We now have all the tools available to prove \Cref{thm:lowrank_frobenius_update}, i.e., the low-rank update to Frobenius forms.

\lowrankfrobenius*

\begin{proof}
    First, observe that
    \[
    (\mA+\mU\mV^\top) = \mK(\mC+\underbrace{\mK^{-1}\mU}_{=:\mU'}\underbrace{\mV^\top\mK}_{=:\mV'^\top})\mK^{-1}.
    \]
    Computing $\mU',\mV'$ takes $O(\MM(n,n,k)) = O(\MM(n,k,k)n/k)$ operations.
    Since genericity is invariant under similarity, the prefix update property is preserved for $\mC$ and the rank-$k$ update~$\mU'\mV'^\top$.
    So we can apply \Cref{cor:companion_update} to compute the Frobenius form of $\mC+\mU'\mV'^\top$ in $\tO(k^{\omega-1}n)$ operations and obtain
    \[
    (\mA+\mU\mV^\top) = \underbrace{\mK\mK'}_{=:\mK''}\mC'\underbrace{\mK'^{-1}\mK^{-1}}_{\mK''^{-1}} = \mK'' \mC' \mK''^{-1}.
    \]
    We can compute $\mK'':=\mK\mK'$ and $\mK''^{-1}=\mK'^{-1}\mK^{-1}$ in $\tilde O(\MM(n,k,k)n/k)$ operations via \Cref{thm:toeplitz-general-matrix-product}. $\mK',\mK'^{-1}$ have displacement rank $O(k)$ and we have their generators by \Cref{cor:companion_update}.
\end{proof}

\section{Dynamic Algorithms}
\label{sec:dynamic}

In this section, we prove our main results: \Cref{thm:dynamic-distances} (dynamic distances) and the algebraic equivalent \Cref{thm:dynamic-powers} (dynamic matrix powers). We also prove the other results from \Cref{sec:intro:results} (\Cref{thm:sensitivity-distances,thm:dynamiccharpoly}).

\dynamicdistances*

We first prove the algebraic result in \Cref{sec:dynamic:algebraic}.
Then in \Cref{sec:dynamic:graph} we prove our graph results by reducing them to our algebraic data structures.

\subsection{Dynamic Matrix Powers and Frobenius Form}
\label{sec:dynamic:algebraic}
We start by proving our dynamic algorithm for maintaining matrix powers via dynamic Frobenius form.

\thmdynamicmatrixpower*

This is $O(n^{1.58})$ update and query time for $k = n^{0.58}$ and bounds from \cite{AlmanDWXXZ25}.

\noindent
Calculation: We write $\MM(n,n^{0.58},n^{0.58})=n^{\omega(1,0.58,0.58)}=n^{\omega(1/0.58, 1, 1)\cdot 0.58}$.
Via convexity of $\omega(\cdot)$ we can bound
$\omega(1/0.58, 1, 1)\le\omega(1.725,1,1) \le 0.45\cdot\omega(2,1,1) + 0.55\cdot\omega(1.5,1,1) \le 0.45\cdot 3.250035 + 0.55\cdot2.794633 \le 3$ where the bounds come from \cite{AlmanDWXXZ25}.
Thus, the update time is $\MM(n,n^{0.58},n^{0.58})n^{1-2~\cdot~0.58} \le n^{3~\cdot~ 0.58+1-2\cdot 0.58} = n^{1.58}$.

\begin{proof}
    We compute the Frobenius form during initialization in $\tO(n^\omega)$ operations~\cite{Keller-Gehrig85}. We maintain it in implicit form
    \[
    \mA = \mK\mT\mC\mT^{-1}\mK^{-1}
    \]
    where $\mK,\mK^{-1}$ are the transforms from the initialization (or reset), $\mT,\mT^{-1}$ are matrices with displacement rank $O(k)$, and we have their generators.
    \paragraph{Update}
    For an entry update that adds $\delta\in\F$ to position $(i,j)$, we have
    \[
    (\mA+e_i e_j^\top\cdot \delta)
    =
    \mK\mT(\mC+\mT^{-1}\mK^{-1}e_i ~\delta~ e_j^\top \mK\mT)\mT^{-1}\mK^{-1}.
    \]
    We can compute $\mT^{-1}\mK^{-1}e_i$ in $\tilde O(nk)$ time because we just take a column of $\mK^{-1}$ and multiply it with a displacement rank $O(k)$ matrix via \Cref{thm:toeplitz-general-matrix-product}.
    The same holds for $e_j^\top \mK\mT$ but taking a row of $\mK$.

    We can now apply \Cref{lem:companion_update} to get $\mC'$ and $\mK',\mK'^{-1}$ in $\tilde O(n)$ operations with
    \[
    (\mA+e_i e_j^\top\cdot \delta)
    =
    \mK\mT\mK'\mC'\mK'^{-1}\mT^{-1}\mK^{-1}.
    \]
    We multiply $\mT\gets\mT\mK'$ and $\mT^{-1}\gets\mK'^{-1}\mT^{-1}$ in $\tilde O(nk)$ via \Cref{thm:toplitz-topelitz-product} because they have displacement rank $O(1)$ and $O(k)$. Then set $\mC \gets \mC'$.
    \paragraph{Query}
    Given $i,j$, we want to compute for all $\ell=1,...,n$
    \[
    \mA^\ell_{i,j}
    =
    e_i^\top \mK \mT \mC^\ell \mT^{-1} \mK^{-1} e_j.
    \]
    Here $e_i^\top \mK$ and $\mK^{-1} e_j$ are a row and a column, respectively, and thus they can be read from memory in $O(n)$ time.
    The product of these vectors with $\mT$ and $\mT^{-1}$ takes $\tilde O(nk)$ time because $\mT,\mT^{-1}$ have displacement rank $O(k)$ and their generators are stored.
    At last, we can compute the product of the vectors with $\mC^\ell$ for all $\ell$ in $\tilde O(n)$ time using~\Cref{lem:companion_powers}.
    In total, the query time is thus $\tilde O(nk)$.

    \paragraph{Reset}
    Every $k$ updates, we reset $\mT\gets\mI$ and use \Cref{thm:lowrank_frobenius_update} to obtain a new Frobenius form, interpreting the past $k$ updates as a rank $k$ update. The genericity requirement of \Cref{thm:lowrank_frobenius_update} is given since those are precisely the states of matrix $\mA$ over the past $k$ updates.
    This takes $\tilde O(\MM(n,k,k)n/k)$ operations.
    The amortized update time is thus
    \[
    \tilde O(\MM(n,k,k)n/k^2 + nk).
    \]
    This can be made worst-case via standard techniques (e.g., \cite{AbrahamCK17,BrandNS19}).
    \paragraph{Space}
    Throughout, the algorithm stored $\mA, \mK, \mK^{-1}, \mC$ and the generators of $\mT,\mT^{-1}$. The explicit matrices need $O(n^2)$ space and the generators use $O(kn)\le O(n^2)$ space.
\end{proof}

\begin{proof}[Proof of \Cref{thm:dynamiccharpoly}]
    Observe that by \Cref{def:frobenius}, the matrix $\mC$ of the Frobenius form corresponds to the characteristic polynomial.
    Thus, the proof above for \Cref{thm:dynamic-powers}, which maintains the Frobenius form, directly implies that we also maintain the characteristic polynomial.
\end{proof}

If we want to support a single batch update rather than a sequence of updates, the following theorem is more efficient. This result will be used for sensitivity distance oracles.

\begin{theorem}\label{thm:sensitivity-powers}
    We can preprocess a generic matrix $\mA\in\F^{n\times n}$ in $\tO(n^\omega)$ operations.
    Then, given a batch of $k\leq n$ entry updates such that applying each prefix of these updates keeps the matrix $\mA$ generic, we can process it in $\tilde O(k^{\omega-1}n)$ operations.
    Finally, we can support queries that w.h.p.~for any $i,j$ in $\tO(nk)$ operations return all $(\mA^\ell)_{i,j}$ for $\ell=1,...,n$.
    The data structure uses $O(n^2)$ space.
\end{theorem}

\begin{proof}
    During initialization we compute the Frobenius form \[\mA=\mK\mC\mK^{-1}.\]
    \paragraph{Batch Update}
    Then, when a batch of $k$ entry updates is given, we encode it using matrices $\mU,\mV\in\F^{n\times k}$ that have only one non-zero entry per column. That is, $\mA+\mU\mV^\top$ represents the batch of $k$ entry updates.

    Next, we have
    \[
    \mA+\mU\mV^\top = \mK(\mC + \mK^{-1}\mU\mV^\top\mK)\mK^{-1}
    \]
    where we can compute $\mK^{-1}\mU$ and $\mV^\top\mK$ in $O(nk)$ time, since they are just $k$ columns of $\mK^{-1}$ and~$k$ rows of $\mK$.
    We compute the Frobenius form of $(\mC + \mK^{-1}\mU\mV^\top\mK)$
    \[
    (\mC + \mK^{-1}\mU\mV^\top\mK)
    =
    \mK' \mC' \mK'^{-1}
    \]
    in $\tilde O(k^{\omega-1}n)$ via \Cref{cor:companion_update}.
    Here, we use the fact that all prefix updates of $\mK^{-1}\mU\mV^\top\mK$ preserve genericity.
    
    \paragraph{Query}
    When receiving $i,j$, we have
    \[
    (\mA+\mU\mV^\top)^\ell_{i,j}
    =
    e_i^\top \mK \mK' \mC'^\ell \mK'^{-1} \mK^{-1} e_j
    \]
    where $e_i^\top \mK \mK'$ and $\mK'^{-1} \mK^{-1} e_j$ are computable in $\tilde O(nk)$ operations, as we just read a row/column of $\mK/\mK^{-1}$ and then multiply with a displacement rank-$k$ matrix via \Cref{thm:toeplitz-general-matrix-product}.
    Then computing the product with $\mC'^\ell$ for all $\ell$ takes $\tO(n)$ operations via \Cref{lem:companion_powers}.
    Thus we obtain $(\mA+\mU\mV^\top)^\ell_{i,j}$ for all $\ell=1,...,n$ in $\tilde O(nk)$ operations.
\end{proof}

\subsection{Dynamic and Sensitivity Distances via Dynamic Matrix Powers}
\label{sec:dynamic:graph}

In this section, we explain how Theorems~\ref{thm:dynamic-distances} and~\ref{thm:sensitivity-distances} follow from Theorems~\ref{thm:dynamic-powers} and~\ref{thm:sensitivity-powers}, respectively.
It is well-known that the powers of a standard adjacency matrix representation of a directed graph encode distances. However, an adjacency matrix is not necessarily generic. To make use of the algebraic algorithms and data structures developed in this paper in distance problems, we refer to the following result of~\cite{KarczmarzS23}.
\begin{theorem}[{\cite[Theorem~7.1]{KarczmarzS23}}]\label{thm:weighted-adjacency-matrix}
Let $G=(V,E)$ be an unweighted directed graph with $n$ vertices.
Let $\mA(G)\in\F^{n\times n}$ be a \emph{weighted adjacency matrix} obtained as follows.
Let \linebreak $\{x_{u,v}:u,v\in V\}\cup \{y_v:v\in V\}$ be a collection of random elements of $\F$, sampled uniformly and independently.
Define:
\[
\mA(G)_{u,v}=\begin{cases}
    x_{u,v}\cdot y_v & \text{ if }u=v\text{ or }uv\in E,\\
    0 & \text{ otherwise.}
\end{cases}
\]
Then, with probability $1-O(n^4/|\F|)$:
\begin{enumerate}[label=(\arabic*)]
\item $\mA(G)$ is generic,
\item For all $u,v\in V$ and $k=0,\ldots,n-1$, $(\mA(G)^k)_{u,v}\neq 0$ if and only if there exists a $u\to v$ path in $G$ of length at most $k$.
\end{enumerate}
\end{theorem}

We can now prove our main dynamic graph results.

\dynamicdistances*

\newcommand{\ds}{\mathcal{D}}

\begin{proof}
Suppose we want to support $q=\poly(n)$ updates to~$G$ so that the produced query answers are all correct with high probability, that is, with probability at least $1-n^{-c}$ for some desired constant $c>0$. 
Let $G_1,\ldots,G_q$ denote the subsequent versions of $G$.

Let the field $\F$ be of $\poly(n)$ size, where the polynomial will be chosen later.
Let a collection of field elements ${X=\{x_{u,v}:u,v\in V\}\cup \{y_v:v\in V\}}$ be sampled as described in~\Cref{thm:weighted-adjacency-matrix}.
Note that as long as the performed updates are independent of the sampled values, the same collection~$X$ guarantees the desired properties~(1)~and~(2) of~\Cref{thm:weighted-adjacency-matrix} for weighted adjacency matrices $\mA(G_i)$ of all graph versions~$G_i$.
That is, the matrices $\mA(G_i)$ obtained from $X$ are all generic, and their powers encode distances in the evolving graph $G$ with probability at least $1-O(qn^4/|\F|)$.
If we pick $\F=\mathbb{Z}/p\mathbb{Z}$ for a prime number $p=\Omega(qn^{4+c})$, we get the desired success probability.

The required independence of updates and sampled values holds for any output-adaptive adversary, since we return the distance which is a deterministically defined property.

The data structure starts by sampling a collection $X$ of~\Cref{thm:weighted-adjacency-matrix}. 
The sampled values are fixed for the entire operation of the data structure.
Then, the weighted adjacency matrix $\mA(G)$ of~\Cref{thm:weighted-adjacency-matrix} is set up using $X$,
and a dynamic matrix powers data structure $\ds$ of~\Cref{thm:dynamic-powers} is initialized.
Observe that updating the matrix $\mA(G)$ upon insertion or deletion of an edge $uv$ requires only changing a single element $\mA(G)_{uv}$.
As a result, updating $\ds$ costs $\tO(nk+\MM(n,k,k)n/k^2)$ time.
By~\Cref{thm:weighted-adjacency-matrix}, an $st$-distance query about the current graph $G$ can be answered, with high probability, by inspecting the values $(\mA(G)^1)_{st},\ldots,(\mA(G)^{n-1})_{st}$.
By~\Cref{thm:dynamic-powers}, these values can be computed using $\ds$ in $\tO(nk)$ time.
Note that the required independence of $X$ and the updates is guaranteed since our graph data structures only reveal distance information encoded in a specific graph version $G_i$ itself, excluding the randomness behind the choice of~$X$, except perhaps in the low-probability event of failure.
\Cref{thm:dynamic-distances} follows.
\end{proof}

Our result for sensitivity distance oracles follows the same overall method, but rather than a sequence of updates, we perform a single batch update.

\thmsensitivity*

\begin{proof}
To obtain~\Cref{thm:sensitivity-distances}, we proceed similarly to the dynamic setting in~\Cref{thm:dynamic-distances}, exploiting the fact that the sample $X$ can be reused for many different graph versions.
We set up a data structure $\ds$ of~\Cref{thm:sensitivity-powers} for the initial matrix~$\mA(G)$.
To enable efficient distance access in the graph $G'$ obtained by applying a batch of $k$ edge updates to $G$, we issue the corresponding 
$k$-element
update to $\ds$.
Note that the prefix-genericity property required by~\Cref{thm:sensitivity-powers} can be guaranteed with high probability analogously to the proof of~\Cref{thm:dynamic-distances}: if we impose an arbitrary order on the $k$ edge updates, we can view the prefix updates as $k$ intermediate graph versions.
Hence, at the cost of $\tO(k^{\omega-1}n)$ batch-preprocessing time, we can access the elements $(\mA(G')^1)_{st},\ldots,(\mA(G')^{n-1})_{st}$ in $\tO(nk)$ time for any $s,t\in V$.
Again,~\Cref{thm:weighted-adjacency-matrix} implies that an $st$-distance query in $G'$ can be answered in $\tO(nk)$ time with high probability.
\end{proof}

\section*{AI Disclosure}

When asked for a rank-$k$ update for the Frobenius form (i.e., as our \Cref{thm:lowrankfrobenius}), GPT-5.2 Pro (February 2026) produced a proof leading to $\tilde O(k^{\omega-2}n^2 + k^\omega n)$ time, making the connection to displacement rank evident.
After we became aware of this connection, we proved the more efficient \Cref{thm:lowrankfrobenius} from scratch and developed the dynamic algorithms (e.g., implicit maintenance of Frobenius form).

\clearpage

\bibliography{ref}
\bibliographystyle{alpha}

\end{document}